\documentclass[11pt]{article}
\usepackage[margin=0.9in]{geometry}

\usepackage{varioref}
\usepackage[usenames,svgnames,xcdraw,table]{xcolor}
\definecolor{DarkBlue}{rgb}{0.1,0.1,0.5}
\definecolor{DarkGreen}{rgb}{0.1,0.5,0.1}

\usepackage{hyperref}
\hypersetup{
   colorlinks   = true,
 linkcolor    = DarkBlue, 
 urlcolor     = DarkBlue, 
	 citecolor    = DarkGreen 
}

\usepackage{appendix}

\usepackage{mathtools}

\usepackage{algorithmic}

\usepackage{array}
\usepackage{adjustbox}
\usepackage[font={small,it}]{caption}

\usepackage{graphicx,epsfig,amsmath,latexsym,amssymb,verbatim}
\usepackage[linesnumbered,ruled,noend, lined]{algorithm2e}
\usepackage{nicefrac}
\let\oldnl\nl
\newcommand{\nonl}{\renewcommand{\nl}{\let\nl\oldnl}}%
\newcommand{\extra}[1]{}

\usepackage{amsmath,amssymb,amsfonts}
\usepackage{amsthm}
\usepackage{thmtools,thm-restate}
\usepackage{enumitem}
\usepackage{tikz}
\usepackage{tikz-cd}
\tikzset{curve/.style={settings={#1},to path={(\tikztostart)
    .. controls ($(\tikztostart)!\pv{pos}!(\tikztotarget)!\pv{height}!270:(\tikztotarget)$)
    and ($(\tikztostart)!1-\pv{pos}!(\tikztotarget)!\pv{height}!270:(\tikztotarget)$)
    .. (\tikztotarget)\tikztonodes}},
    settings/.code={\tikzset{quiver/.cd,#1}
        \def\pv##1{\pgfkeysvalueof{/tikz/quiver/##1}}},
    quiver/.cd,pos/.initial=0.35,height/.initial=0}
\newtheorem{theorem}{Theorem}

\newtheorem{corollary}{Corollary}

\newtheorem{definition}{Definition}
\newtheorem{example}{Example}
\newtheorem{lemma}{Lemma}

\newtheorem{claim}{Claim}

\def\squareforqed{\hbox{\rlap{$\sqcap$}$\sqcup$}}
\def\qed{\ifmmode\squareforqed\else{\unskip\nobreak\hfil
\penalty50\hskip1em\null\nobreak\hfil\squareforqed
\parfillskip=0pt\finalhyphendemerits=0\endgraf}\fi}
\def\endenv{\ifmmode\;\else{\unskip\nobreak\hfil
\penalty50\hskip1em\null\nobreak\hfil\;
\parfillskip=0pt\finalhyphendemerits=0\endgraf}\fi}

\renewenvironment{proof}{\noindent \textbf{{Proof~} }}{\qed\medskip}
\newenvironment{proof+}[1]{\noindent \textbf{{Proof #1~} }}{\qed\medskip}

\mathchardef\ordinarycolon\mathcode`\:
\mathcode`\:=\string"8000
\def\vcentcolon{\mathrel{\mathop\ordinarycolon}}
\begingroup \catcode`\:=\active
  \lowercase{\endgroup
  \let :\vcentcolon
  }

\DeclareMathOperator*{\argmin}{arg\,min}
\DeclareMathOperator*{\argmax}{arg\,max}

\newcommand{\red}[1]{{\leavevmode\color{red}{#1}}}

\newcommand\toref{\red{[REF]}}
\newcommand{\cut}{\texttt{cut}}
\newcommand{\eval}{\texttt{eval}}
\newcommand{\dTwoTree}{\textsc{Depth2Tree}}
\newcommand{\twoStar}{2-\textsc{Star}}

\newcommand{\AlgOne}{\textsc{Alg1}}

\newcommand{\Inact}{\texttt{Inact}}
\newcommand{\Inchild}{\texttt{Inchild}}

\newcommand{\N}{\mathcal{N}}

\tikzcdset{every arrow/.append style = -{Latex[scale=2.5]}}
\DeclarePairedDelimiter\floor{\lfloor}{\rfloor}

\definecolor{DarkBlue}{rgb}{0.1,0.1,0.5}
\definecolor{DarkGreen}{rgb}{0.1,0.5,0.1}
\usepackage{hyperref}
\usepackage{mathrsfs}
\hypersetup{
   colorlinks   = true,
 linkcolor    = DarkBlue, 
 urlcolor     = DarkBlue, 
	 citecolor    = DarkGreen 
}

\newcommand{\DominationCondition}{\textsc{Fair}}
\newcommand{\Storage}{\texttt{Storage}}
\newcommand{\DominationProtocol}{$\textsc{Domination}(R,k)$}
\newtheorem{observation}{Observation}

\usepackage[normalem]{ulem}
\usepackage[T1]{fontenc}
\usepackage{lmodern, microtype}
\usepackage[small]{titlesec}
\usepackage{csquotes}
\usepackage{comment}
\renewcommand{\mkbegdispquote}[2]{\itshape}
\usepackage{wrapfig}
\usepackage{arydshln}

\usepackage{tikz}
\usepackage{tikz-cd}
\tikzset{curve/.style={settings={#1},to path={(\tikztostart)
    .. controls ($(\tikztostart)!\pv{pos}!(\tikztotarget)!\pv{height}!270:(\tikztotarget)$)
    and ($(\tikztostart)!1-\pv{pos}!(\tikztotarget)!\pv{height}!270:(\tikztotarget)$)
    .. (\tikztotarget)\tikztonodes}},
    settings/.code={\tikzset{quiver/.cd,#1}
        \def\pv##1{\pgfkeysvalueof{/tikz/quiver/##1}}},
    quiver/.cd,pos/.initial=0.35,height/.initial=0}
\usetikzlibrary{arrows}
\tikzcdset{every arrow/.append style = -{Latex[scale=2.5]}}

\newcommand{\yuka}[1]{{\color{blue}{Yuka: #1}}}

\date{}

\begin{document}
\title{A Discrete and Bounded Locally Envy-Free Cake  Cutting Protocol on Trees}
\author{Ganesh Ghalme\thanks{Indian Institute of Technology, Hyderabad, India.  {\tt ganeshghalme@ai.iith.ac.in} }  \and Xin Huang\thanks{Technion Israel Institute of Technology, Haifa, Israel. {\tt xinhuang@campus.technion.ac.il}} \and Yuka Machino\thanks{Massachusetts Institute of Technology, USA. {\tt yukam997@mit.edu}} \and Nidhi Rathi\thanks{Aarhus University, Denmark. {\tt nidhi@cs.au.dk}}}

\maketitle   
\begin{abstract}
We study  the classic problem of \emph{fairly} dividing a heterogeneous and divisible resource---modeled as a line segment $[0,1]$ and typically called as a \emph{cake}---among $n$ agents. This work  considers  
an interesting variant of the problem where agents are embedded on a graph. The graphical constraint entails that each agent evaluates her allocated share only against her neighbors' share. Given a graph, the goal is to efficiently find a \emph{locally envy-free} allocation where every agent values her share of the cake  to be at least as much as that of any of her neighbors' share.

The most significant contribution of this work is a bounded protocol that finds a locally envy-free allocation among $n$ agents  on a  tree graph using $n^{O(n)}$ queries under the standard Robertson-Webb (RW) query model. The  query complexity   of our proposed protocol, though exponential, significantly improves the currently  best known hyper-exponential query complexity bound of Aziz and Mackenzie \cite{aziz2016discrete} for complete graphs. In particular,  we also show that if the underlying tree graph has a depth of  at most two,  one can find a locally envy-free allocation  with $O(n^4 \log n)$ RW queries. This is the first and the only known locally envy-free cake  cutting protocol with polynomial query complexity for a non-trivial graph structure.

Interestingly, our discrete protocols are simple and easy to understand, as opposed to highly involved protocol of \cite{aziz2016discrete}. This simplicity can be attributed to their  recursive nature and the use of a single agent as a designated \emph{cutter}.  We believe that these results will help us improve our algorithmic understanding of the arguably challenging problem of envy-free cake-cutting  by uncovering the  bottlenecks in its query complexity and its relation to the underlying graph structures. 

\end{abstract} 


\section{Introduction}
The problem of fairly dividing resources among a set of participating agents is one of the fundamental problems whose roots date back to 1950's \cite{steihaus1948problem,dubins1961cut} while enjoying widely-ranged real-world applications \cite{gal2016fairest, moulin2004fair,vossen2002fair,etkin2007spectrum,budish2011combinatorial}. Over the past several decades, fair division has been extensively studied across various disciplines like social science, economics, mathematics and computer science; see \cite{brams1996fair,brandt2016handbook,procaccia2015cake,robertson1998cake} for excellent expositions. Among the resource-allocation settings, the \emph{cake-cutting problem} provides an elegant mathematical abstraction to many real world situations where a divisible resource---modeled as a \emph{cake} $[0,1]$---is to be allocated among agents with heterogeneous preferences. These situations include   division of land, allocation of radio and television   spectrum,  allocation of advertisement space on search platforms and so on (see \cite{AdjustedWinner} for implementations of cake-cutting methods). Cake-cutting has indeed been the basis of the key axiomatic formalization of fair division along-with the inspiration of some of the central solution concepts for resource allocation.

Formally, a cake-division instance consists of $n$ agents having cardinal preferences over the cake, that is modeled as a unit interval $[0,1]$. These preferences are specified by valuation functions $v_i$'s, and we write $v_i(I)$ to denote agent $i$'s value for the piece $I \subseteq [0,1]$. The goal is to partition the cake into $n$ bundles (possibly consisting of finitely many intervals) and assign them to the $n$ agents.

A central notion of fairness in resource-allocation settings is that of \emph{envy-freeness} that deems a cake division to be \emph{fair} if every agent prefers her share over that of any other agent \cite{foley1967resource}. That is, an allocation $\{A_1, A_2, \dots, A_n\}$ is said to be \emph{envy-free} if and only if $v_i(A_i) \geq v_i(A_j)$ for all $i,j \in [n]$. The appeal of envy-freeness can be rightfully perceived from the strong existential results: under mild assumptions, a cake-division (where every agent receives a connected interval) is always guaranteed to exist \cite{stromquist1980cut,simmons1980private,edward1999rental}. While this compelling existential result is derived from elegant mathematical connections with the area of topology, the corresponding efficient algorithmic results remain elusive. As a matter of fact, Stromquist \cite{stromquist2008} proved that there cannot exist a finite protocol for computing an envy-free cake division with connected pieces (in an adversarial model). Furthermore, Deng et al. \cite{deng2012algorithmic} showed that the problem of finding envy-free division (with connected pieces) under ordinal valuations is PPAD-hard.

Even though the only known lower bound on the query complexity of finding envy-free cake divisions is $\Omega(n^2)$ \cite{Procaccia2009ThouSC}, the best known algorithm \cite{aziz2016discrete} has a hyper-exponential query complexity bound of $n^{n^{n^{n^{n^{n}}}}}$ for finding envy-free cake division with non-contiguous pieces. This evidently leaves a huge gap in our understanding of the query complexity of the underlying problem.  


 In this work, we  partially  address this gap by
 exploring the graphical framework of envy-freeness in cake cutting. Here, the  envy comparisons are restricted by an underlying (social) graph $G$ over the agents.  In contrast with the standard notion of envy-freeness, the goal here is to find a \emph{locally envy-free} allocation of the cake such that no agent envies her neighbor(s) in the given  graph $G$. That is, an allocation $\{A_1, A_2, \dots, A_n\}$ is said to be \emph{locally envy-free} if and only if for all $i \in [n]$, we have $v_i(A_i) \geq v_i(A_j)$ for all $j \in N_i(G)$, where $N_i(G)$ is the set of neighbours of agent $i$ in the underlying graph $G$. Note that  when $G$ is a complete graph, we retrieve the classical setting of envy-free cake division. 
 

The above-described graphical framework opens various interesting directions for understanding the problem of fairness in cake division (see \cite{AbebeKP16,bei2017networked,bei2020cake,Bredereck18,Tucker21}). The notion of local envy-freeness  is  relevant in many natural scenarios where agents' envy towards other agents can be restricted by  external constraints such as social connections, relative rank hierarchy and  overlap in their expertise/skill levels.  For instance, when the graph represents social connections between  a  group of people, it is reasonable to assume that agents  only envy the agents whom they know (i.e., friends or friends of friends). Similarly, when a graph represents rank hierarchy in an organization,  it is reasonable to assume that agents only envy their immediate neighbours (i.e. colleagues). 
The study of \emph{local envy-freeness} is particularly interesting from a purely theoretical standpoint. Given that the state of the art protocols for finding envy-free divisions are highly complicated and require hyper-exponentially many queries, a natural line of research---and the focus of this work---is to find interesting graph structures for which the problem of fair cake division admits  query-efficient cake cutting protocols. 

In the light of the algorithmic barriers, identification of graph structures for better query complexity stands as a meaningful direction of work. We believe that non-trivial graph structures provide necessary building blocks and help understand the bottleneck of computational difficulties in the problem of fair cake division.



\subsection{Our Results and Techniques}  All the protocols developed in this paper operate under the Robertson-Webb query model that allows access to agents valuations via \cut\ and \eval\ queries (see Section \ref{sec:setting} for details).  The underlying graph over the agents is undirected i.e., an edge between two agents in the graph correspond to the envy-constraint where neither of the agents envy another.

\begin{itemize}

\item \emph{Local envy-freeness on a \textsc{Line}:} We design a simple discrete protocol to compute a locally envy-free allocation among $n$ agents on a \textsc{Line} (represented as $a_1-a_2- \dots -a_n$) using $n^{(O(n)}$ queries.

At the heart of our algorithm lies a recursive protocol $\textsc{Domination}(R,k)$ (for some piece $R \subseteq[0,1]$ and agent $a_k$ for $k\in [n]$). In particular, $\textsc{Domination}(R,k)$ repeatedly invokes $\textsc{Domination}(R,k+1)$ until a certain {\em domination condition} is achieved.\footnote{ $\textsc{Domination}(R,n)$ asks agent $a_n$ to cut $R$ into $n$ equal pieces according to her.} 
The  \DominationProtocol\ protocol returns an $n$-partition of $R$ satisfying the following three key conditions (for the allocation achieved by assigning the $j$th bundle to agent $a_j$): (a) agents $a_i$ for $i \geq k$ do not envy their neighbours, (b) agent $a_k$ values the first $k$ bundles equally, and (c) agent $a_{k+1}$ values her bundle at least as high as the first $k$ bundles. We formalize these crucial conditions and call such an allocation as $k$-\DominationCondition. This notion of $k$-\DominationCondition ness proves to be the stepping stone into going ahead and beating the highly complicated state-of-the-art envy-free cake-cutting protocols.
Finally, note that the allocation returned by $\textsc{Domination}([0,1],1)$ is $1$-\DominationCondition, i.e., a locally envy-free on a \textsc{Line}.

The above idea of recursion imparts notable simplicity to our protocol. The number of times \DominationProtocol\ calls $\textsc{Domination}(R,k+1)$ determines the query complexity of the protocol, and we show that it is at most $k+k \log(k)$ calls. Overall, we establish that our protocol $\textsc{Domination}([0,1],1)$ requires $n^{O(n)}$ queries to find a locally envy-free allocation among $n$~agents on a \textsc{Line}.

\item \emph{Local envy-freeness on a \textsc{Tree}:} We prove our most important result in the context of \textsc{Tree} graphs, and develop a protocol for finding a locally envy-free allocation among $n$ agents on \emph{any} \textsc{Tree} in $n^{(O(n)}$ queries. We achieve so by performing appropriate modifications in the protocol for {\sc Line} graph (without increasing the query complexity) and carefully adapting the important notion of $k$-\DominationCondition\ allocations in the context of \textsc{Tree} graphs.

It is relevant to note that our protocol is a significant improvement from the known hyper-exponential query complexity bound of Aziz and Mackenzie \cite{aziz2016discrete} for finding envy-free allocations among $n$ agents on complete graphs. Interestingly, we are able to efficiently generalize our key techniques developed for \textsc{Line} graphs and hence we maintain the identical query complexity bound of $n^{(1 + o(1))n}$ as that of \textsc{Trees}. The bottleneck for the query complexity is determined by the longest path on the given \textsc{Tree}, and hence, the {\sc Line} graph proves to be the hardest graph structure among trees for our proposed protocol (with respect to the query complexity).  

It is relevant to note that if our protocol $\textsc{Domination}([0,1],1)$ for \textsc{Trees} is run on a \textsc{Star} graph, then it indeed finds a locally envy-free allocation using $n^2$ queries.\footnote{Consider the following protocol for a \textsc{Star} graph: the center agent cuts the cake $[0,1]$ into $n$ equal pieces, and the remaining agents picks their favorite available piece (in any arbitrary order), with the last remaining piece being assigned to the center agent. This protocol find a locally envy-free allocation using $n^2$ queries.}

\item \emph{Polynomial-query protocol for local envy-freeness on \dTwoTree:}
 We identify two non-trivial graph sub-structures of \textsc{Trees}, namely \dTwoTree\ (defined as the tree graph with depth at most two) and \twoStar\ (defined as \dTwoTree\ where every non-root agent in connected to at most two agents) that admits query-efficient protocols for finding locally envy-free allocation. In particular, we design a protocol $\textsc{Alg2}$ that uses $O(n^4 \log(n))$ queries for \dTwoTree\ and $O(n^3)$ queries \twoStar\ to compute locally envy-free allocations. 
 
 As a warm-up, we first design a protocol (\textsc{Alg1}) for four agents on a $\textsc{Line}$ and then generalize the techniques to develop the protocol \textsc{Alg2} for the above-mentioned graph structures. Note that $\textsc{Domination}([0,1],1)$ protocol for \textsc{Trees} gives exponential query complexity for these simple graphs. Hence, we optimize our techniques and develop new ideas to obtain query-efficiency for \dTwoTree\ and \twoStar\ via our non-recursive protocol \textsc{Alg2}.
 \end{itemize}

\subsection{Related Literature:}
Fairness in resource-allocation settings is extensively studied in economics, mathematics and computer science literature (see \cite{procaccia_moulin_2016,brams1996fair,moulin2004fair}).  While strong existential guarantees are known for envy-free cake divisions \cite{stromquist1980cut,edward1999rental}, the corresponding computational problem remain challenging \cite{deng2012algorithmic,stromquist2008}. 
For three agents,\footnote{For two agents a simple cut-and-choose protocol returns the envy free allocation with a single cut on the cake.} the celebrated Selfridge-Conway protocol \cite{robertson1998cake} finds an envy-free allocation using 5 \cut\ queries. However, despite significant efforts, developing efficient envy-free cake cutting protocols for $n$ agents remains largely open: the current known upper bound has hyper-exponential dependency on the number of agents \cite{aziz2016discrete},  whereas  the only known lower bound for the problem is $\Omega(n^2)$ \cite{Procaccia2009ThouSC}. This leaves an embarrassingly huge gap in our computational understanding of the problem. Attempts have been made to address this gap with various kind of approaches. Aziz and Mackenzie \cite{aziz2016discretefour} proposed a cake cutting protocol that finds an envy-free allocation among four agents in (close to) $600$ queries. This bound was recently improved   by \cite{AmanatidisFMPV18} to $171$ queries. On the other hand, efficient fair cake cutting protocols for interesting classes of valuations have been developed in \cite{kurokawa2013cut,barman2021fair}. Furthermore, Barman et al. \cite{arunachaleswaran2019fair} developed an efficient algorithm that finds a cake division (with connected pieces) wherein the envy is multiplicatively within a factor of $2+o(1/n)$. 

 
The problem of cake cutting with graphical (envy) constraints was first introduced  by Abebe et al. \cite{AbebeKP16}. They fully characterize the set of (directed) graphs for which an oblivious single-cutter protocol---a protocol that uses a single agent to cut the cake into pieces---admits a bounded protocol for locally envy-free allocations in the Robertson-Webb model. In contrast, our work studies a class of \emph{undirected} graphs that are significantly harder to analyze, and surprisingly develops comparable upper bounds. In another closely related paper, Bei et al. \cite{bei2017networked} develops a moving-knife protocol\footnote{A moving-knife protocol may not be implementable in discrete steps under the standard Robertson-Webb query model, and hence is considered weaker.} that outputs an envy-free allocation on tree graphs. In contrast, we develop a bounded protocol under the standard Robertson-Webb query model for finding envy-free allocations on trees. In a more recent work,  Bei et al. \cite{bei2020cake} develop  a discrete and bounded {\em locally proportional}\footnote{For a given graph, an allocation is said to be locally proportional if every agent values her share that is at least as high as her average value of her neighbours' shares.} protocol for any given graph. In contrast, our work addresses stronger guarantee of local envy-freeness. We address the open question raised in \cite{bei2020cake} by (a) developing a discrete and bounded protocol for trees, and  (b) constructing a query-efficient discrete protocol that finds  locally envy-free allocations among $n$~agents on tree graphs with depth at most two. 

Local envy-freeness has also been explored in the discrete setting \cite{beynier2019local, chevaleyre2017distributed} where there is a set of indivisible items, and every agent is assigned a bundle of items. In the discrete setting, Aziz et al. \cite{aziz2018knowledge} defines new class of fairness notions parameterized by an underlying social graph and position them  with respect to known ones, thus revealing new rich hierarchies of fairness concepts. The work of \cite{bredereck2022envy} studies the parameterized computational complexity with respect to a few natural parameters such as the number of agents, the number of items, and the maximum number of neighbors of an agent.

\section{The Setting }
\label{sec:setting}
We consider the problem of fairly dividing a heterogeneous and divisible resource---modeled as a unit interval and referred as a \emph{cake}---among $n$ agents, denoted by the set $\N = \{a_1, a_2, \cdots , a_n\}$. For an agent $a_i \in \N$, we write  $v_i$ to specify her (cardinal) valuations over the intervals in $[0,1]$. In particular, $v_i(I) \in \mathbb{R}^+ \cup \{0\}$ represents the valuation of agent $a_i$ for the interval $I \subseteq [0,1]$. For brevity, we will write $v_i(x,y)$ instead of $v_i([x,y])$ to denote agent $a_i$'s value for an interval $[x,y]\subseteq [0,1]$.  Following the standard convention, we assume that $v_i$s are non-negative, additive,\footnote{For any two disjoint intervals $I_1, I_2 \subseteq [0,1]$, we have $v_{i}(I_1 \cup I_2) = v_{i}(I_1) + v_{i}(I_2)$.} and non-atomic.\footnote{For any interval $[x,y] \subseteq [0,1]$ and any $\lambda \in [0,1]$, there exists a $y'$ such that $v_i(x,y') = \lambda \cdot v_{i}(x,y)$.} Additionally, without loss of generality, we assume that the valuations are
normalized i.e., we have $v_i(0,1) = 1$ for all $i \in [n]$. 

We write  $G := (V,E)$ to denote the underlying social graph over the agents. Here, vertex $i \in V$ represents agent $a_i \in \N$ and an (undirected) edge $(i,j) \in E$ represents a connection (or the envy-constraint) between agents $a_i$ and $a_j$. An edge $(i,j)$ would refer to the constraint that agents $a_i$ and $a_j$ should not envy each other.

\subsection{Preliminaries} \label{section:preliminaries}

\noindent
\textbf{Problem instance:} A  \emph{cake-division instance $\mathcal{I}$ with graph constraints} is denoted by a tuple $\langle \N, G, \{v_i\}_{i \in [n]} \rangle$. Here, $\N$ denotes the set of $n$ agents, $G$ represents the social graph over the agents and $v_i$s specify the valuations of agents over the cake $[0,1]$.\\

\noindent
\textbf{Allocations:} For cake-division instances, we define an \emph{allocation} $\mathcal{A} := \{ A_1, A_2, \cdots , A_n\}$ of the cake $[0,1]$ to be a collection of $n$ pair-wise disjoint pieces such that $\cup_{i \in [n]} A_i = [0,1]$. Here, a piece or a bundle $A_i$ (a finite union of intervals of the cake $[0,1]$) is assigned to agent $a_i \in \N$. We say $\mathcal{A}$ is a \emph{partial} allocation if the union of $A_i$s forms a strict subset of $[0,1]$. \\

In this work, we study protocols for finding \emph{locally envy-free} allocations, a natural extension of the well studied notion of  envy-freeness defined below. 
\begin{definition}[Envy-freeness]
 For a cake-division instance, an allocation $\mathcal{A}$ is said to be \emph{envy-free} if we have $v_i(A_i) \geq v_i(A_j)$ for all agents $i,j \in [n]$.
\end{definition}


\begin{definition}[Local Envy-freeness]
  Given a  cake-division instance with a social graph $G=(\mathcal{N},E)$, an allocation $\mathcal{A} := \{ A_1, A_2, \cdots , A_n\}$ is said to be \emph{locally envy-free} (on $G$) if for all agents $i \in [n]$, we have  $v_{i}(A_i) \geq v_{i}(A_j) $ for all $j \in \mathcal{N}$ such that $(i,j) \in E$. 

\end{definition}

Local envy-freeness ensures that every agent prefers her own piece over that of any of her neighbors in $G$. When $G$ is the complete graph over $n$ agents, we recover the classical fairness guarantee of envy-freeness.  
Our algorithms operate under the Robertson-Webb model \cite{robertson1998cake} defined below. 

\begin{definition}[Robertson-Webb query model]
  Our protocols access the agents' valuations by making the following two types of queries:\begin{enumerate}
      \item  Cut query:  Given a point $x \in [0,1]$ and a target value $\tau \in [0,1]$, $\cut_i(x,\tau)$ asks agent $a_i$ to report the subset $[x,y]$ such that $v_i(x,y) = \tau$. If such a subset does not exist, then the query response is some pre-determined garbage value.
\item Evaluation query: Given $0 \leq x < y \leq 1 $, $\eval_{i} (x,y)$ asks agent $a_i$ to report her value $ v_i (x,y)$ for the interval $[x,y]$ of the cake. 
  \end{enumerate}
\end{definition}

\noindent
\textit{Remark:} A cut query $\cut_i(x,\tau)$ which takes in an input point $x \in [0,1]$ and a value $\tau$ can easily be translated into a query $\cut_i(\mathcal{X},\tau)$ that takes in an input a collection of $k$ intervals (set according to left to right occurence) $\mathcal{X}=\{X_1,X_2\ldots X_k\}$ and outputs a collection of intervals $\Tilde{\mathcal{X}}=\{X_1,X_2,\dots, X_{k'-1}, \Tilde{X}_{k'}\}$ where $v_i(\Tilde{\mathcal{X}})=\tau$.\\
    This can be done since $a_i$ knows $v_i(X)$ for all $X\in \mathcal{X}$, and hence it can identify the $k' \in [k]$ such that $v_i(X_1+X_2\cdots X_{k'-1})< \tau \le v_i(X_1+X_2\dots X_{k'})$. Agent $a_i$ lets $\Tilde X_{k'}=\cut_i(x,\tau')$ where $\tau'=\tau-v_i(X_1+X_2\cdots X_{k'-1})$ and where $x$ is the smaller end of the interval of $X_{k'}$ to obtain $\Tilde{\mathcal{X}}=\{X_1,X_2\ldots X_{k'-1}, \Tilde{X_{k'}}\}$ with $v_i(\Tilde{X}_{k'})=\tau'$. \\

\noindent
\textbf{Special graph structures:}
Our work focuses on instances where the underlying social graph over the agents is either a \textsc{Tree}, a \textsc{Line}, a \dTwoTree, or a \twoStar. 
We say that a graph is \dTwoTree \ if it is a tree with depth at most two. A \twoStar\ graph is a special case of \dTwoTree\ where each non-root agent is connected to at most two  agents.

 \paragraph{Terminology and important subroutines:} 
 We now state the terminology that is used to describe our protocols. An agent is called a \emph{cutter} if she makes cuts on the unallocated piece of the cake. Interestingly, the proposed protocols in this work require only a single agent to to act as a cutter. Furthermore, our protocol always require the cutter agent to makes cuts to divide a certain piece of the cake into $n$ equal parts (according to her). An agent is called a \emph{trimmer} agent if she performs either a \textsc{Trim(.)} or \textsc{Equal(.)} procedure, as defined below. \\  


  \noindent 
\textbf{Select:} Given a collection of pieces  $\mathcal{X}$, $m \leq  |\mathcal{X}|$,  and an agent $a_i$, $\textsc{Select}(\mathcal{X},a_i,m)$  returns  the $m$ largest valued pieces according to $v_i$. It is easy to see that \textsc{Select} requires zero \texttt{cut}  queries and at most  $|\mathcal{X}|$ \texttt{eval}  queries.\\

\noindent 
  \textbf{Trim}: Given a collection of pieces  $\mathcal{X}$  and an agent $a_i$, $\textsc{Trim}( \mathcal{X}, a_i)$ returns a collection of $|\mathcal{X}|$ pieces where each piece is valued equal to her smallest-valued piece in $\mathcal{X}$ and some \emph{residue}. The procedure first finds the lowest valued piece according to $v_i$ and makes the remaining pieces of value equal to it by trimming. Trimmings (denoted by $U$ in Steps 5 and 6 of the procedure) are returned separately, which we call as \emph{residue} of the procedure. The \textsc{Trim} procedure requires $|\mathcal{X}|-1$ \texttt{cut} queries and  $|\mathcal{X}|$  \texttt{eval} queries. \\    

\noindent 
\textbf{Equal}: Given a collection of pieces  $\mathcal{X}$  and an agent $a_i$, $\textsc{Equal}(\mathcal{X}, a_i)$ redistributes among the pieces in $\mathcal{X}$ such that each piece is equally valued by $a_i$. It also identifies a bundle in the original collection that has value larger than the average value of the bundles (according to $v_i$). Note that while both  \textsc{Equal} and \textsc{Trim} procedures return an allocation where all the pieces are equally valued by $a_i$, \textsc{Trim} may generate a residue whereas \textsc{Equal} procedure redistributes all the cake into the same number of pieces without leaving any part unallocated. The \textsc{Equal} procedure requires $|\mathcal{X}| -1 $ \texttt{cut} queries and $|\mathcal{X}|$ \texttt{eval} queries.  

    \begin{minipage}{0.46\linewidth}
            \vspace{-2pt}
\begin{algorithm}[H]
    \renewcommand{\thealgocf}{}
    \DontPrintSemicolon
    \SetAlgorithmName{$\textsc{Proc}$}{ }{ }
    \caption{\textsc{Select}($a_i, \mathcal{X}, m$)}
        \label{pro:trim}
        Initialize $\mathcal{X}^i \gets \emptyset$\;
        \For{$j = 1 \to m$}
        {Update $\mathcal{X}^{(i)} \gets \mathcal{X}^{(i)} \cup \argmax_{X \in \mathcal{X}} v_{i}(X)$, and $\mathcal{X} \gets \mathcal{X} \setminus \mathcal{X}^{(i)}$\;
      }
     \Return $(\mathcal{X}^{(i)}, \mathcal{X})$
\end{algorithm} 
\vspace{10pt}
\begin{algorithm}[H]
    \renewcommand{\thealgocf}{}
    \DontPrintSemicolon
    \SetAlgorithmName{$\textsc{Proc}$}{ }{ }
    \caption{\textsc{Trim}($a_i, \mathcal{X}$)}
        \label{pro:trim}
        \nonl Let $\mathcal{X} = \{X_0, X_1, \dots, X_{\ell_i}\}$, s.t. $X_0=\argmin_{X \in \mathcal{X}}{v_{i}(X)}$\;
        Set $\tau = v_{i}(X_0)$ \;
        Initialize $U \gets \emptyset$ \;
    \For{$j =1 \to \ell_i$}
     {Let $\cut_i(X_j,\tau) = X'_j$ with $v_{i}(X'_j)= v_{i}(X_{0})$, and let $T_j=X_j \setminus X$\;
     Update $X_j \gets X'_j$ and $U \gets U \cup T_j$\;
     }
     \Return $(\mathcal{X},U)$
\end{algorithm} 
    \end{minipage}
   \medskip
    \vspace{10pt} 
         \hspace{5pt}
    \vline 
    \hspace{5pt} 
    \begin{minipage}{0.47\linewidth}
\begin{algorithm}[H]
    \renewcommand{\thealgocf}{}
    \DontPrintSemicolon
    \SetAlgorithmName{$\textsc{Proc}$}{ }{ }
    \caption{\textsc{Equal}($a_i, \mathcal{X}$)}
        \label{pro:equal}
    Let $\mathcal{X}^{s} := \{X \in \mathcal{X}: v_{i}(X) < v_{i}(\mathcal{X})/|\mathcal{X}|\}$,  $\mathcal{X}^{\ell} := \{X \in \mathcal{X}: v_{i}(X) \geq v_{i}(\mathcal{X})/|\mathcal{X}|\}$\;
    Initialize $\mathcal{T} \gets \emptyset$ and let $\tau = v_{i}(\mathcal{X})/|\mathcal{X}|$\; Let $X^*$  be  arbitrary element of $\mathcal{X}^{\ell}$\;
    \For{$X_j \in \mathcal{X}^{\ell}$}
    {
    $a_i$ divides $X_j$ into $X_j'$ and $T_j$ such that $v_{i}(X_j') = \tau$ using a single cut\;
    $\mathcal{T} \gets \mathcal{T} \cup T_j$ and $X_j \gets X'_j$
    }
      \For{$X_j\in \mathcal{X}^{s}$}
        {
            \While{$v_i(X_j)< \tau$}
            {
                $X_j\leftarrow X_j\cup T$ for some $T \in \mathcal{T}$\;
                $\mathcal{T} \leftarrow \mathcal{T}\setminus T$\;
            }
            Let $T' \in \mathcal{T}$ be  last piece added to $X_j$\;
            Cut $T'$ into two pieces $T'_{1}$ and $T'_{2}$ such that $v_i(X_j')=\tau$ where $X_j'=X_j\setminus T'_{2}$\;
            $T\leftarrow T\cup T'_{2}$ and $X_j \gets X'_j$\;
        }
     \Return $(\mathcal{X}, X^*)$;\;
\end{algorithm} 
    \end{minipage}




\section{Main Results}
This section presents the statements of our key results.  We begin by developing a simple recursive  protocol in Section~\ref{sec:line} that finds a locally envy-free allocation among $n$ agents on a \textsc{Line} using $n^{O(n)}$ queries.
  \begin{restatable}{theorem}{nAgentsOnLine} \label{thm:line}
 For cake-division instances with $n$ agents on a \textsc{Line}, there exists a discrete cake-cutting protocol that computes a locally envy-free allocation using at most $n^{(1+o(1))n}$ queries in the Robertson-Webb query model.
\end{restatable}
 
Next, in Section~\ref{sec:tree} we carefully adapt the discrete protocol developed for the case of \textsc{Line} and modify it to construct a discrete protocol for computing locally envy-free allocation among $n$ agents on a \textsc{Tree} using $n^{O(n)}$ queries. As mentioned earlier,  we are able to efficiently generalize our key techniques developed for \textsc{Line} graphs and hence we maintain the identical query complexity bound. Designing a discrete and bounded protocol for local envy-freeness on trees is listed as an open problem in \cite{AbebeKP16}. We answer this open problem by proving the following result.
 \begin{restatable}{theorem}{nAgentsOnTree} \label{thm:tree}
 For cake-division instances with $n$ agents on a  \textsc{Tree}, there exists a discrete cake-cutting protocol that computes a locally envy-free allocation using at most $n^{(1+o(1))n}$ queries in the Robertson-Webb query model.
\end{restatable}

We remark here that the above result for \textsc{Tree} graphs, though requires exponentially many queries, is a significant improvement over the best known hyper-exponential query complexity bound for complete graphs \cite{aziz2016discrete}.\footnote{We remark here that any locally envy-free cake-cutting protocol on a graph $G$ is also locally envy-free on any subgraph $G' \subseteq G$.} 
We believe that it can prove to be a stepping stone in understanding the bottlenecks of computational complexity of envy-free allocations.

Next, in Section~\ref{sec:special graphs}, we turn our attention towards identifying specific class of graph structures that admit polynomial-query algorithms for local envy-freeness. Developing efficient locally envy-free  protocols for interesting graph structures is listed as an open problem in \cite{bei2020cake} and \cite{AbebeKP16}. This work partially addresses  the aforementioned open problem by developing  a novel protocol that finds a locally envy-free allocation using  polynomial queries for \dTwoTree\ and \twoStar\ graph structures. 

\begin{restatable}{theorem}{Trees}\label{thm:depth2tree}
\label{thm:tree}
  For cake-division instances with $n$ agents on a \textsc{Depth2Tree}, there exists a discrete protocol  that finds a locally envy-free allocation using at most $O(n^3\log(n))$ \cut \ and $O(n^4\log(n))$ \eval \ queries.
\end{restatable}

 \begin{restatable}{theorem}{TwoStar}\label{thm:2star}
 For cake-division instances with $n$ agents on a $2$-\textsc{Star}, there exists a discrete protocol that finds a locally envy-free allocation using at most $O(n^2)$ \cut \ and $O(n^3)$ \eval \ queries.
\end{restatable}

\section{Local Envy-freeness on $\textsc{Line}$ graphs}\label{sec:line}
In this section, we will study the cake-division instances where envy-constraints among agents are specified by a \textsc{Line} graph, i.e., agents lie on the social graph of \textsc{Line}, as shown in Figure \ref{fig: n agents on line}. We present one of our key results that develops a bounded protocol $\textsc{Domination}([0,1],1)$ for finding a locally envy-free allocation among $n$ agents on a \textsc{Line} using $n^{(1+o(1))n}$ queries in the Robertson-Webb query model. Note that, this protocol is of significantly lower complexity as compared to the best known hyper-exponential query bound for complete graphs \cite{aziz2016discrete}.

 \begin{figure}[ht!]
     \centering
     \includegraphics[scale=0.9]{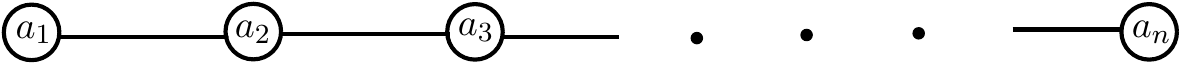}
     \caption{$n$ agents on a \textsc{Line} graph}
     \label{fig: n agents on line}
 \end{figure}

Our protocol $\textsc{Domination}([0,1],1)$ is recursive in nature where for any piece $R \subseteq [0,1]$ and any agent $a_k$ for $k \in [n-1]$, \DominationProtocol\ invokes $\textsc{Domination}(R,k+1)$ multiple times. Before giving an overview of the protocol \DominationProtocol, we begin by describing the following crucial property, to be referred as $k$-\DominationCondition ness, that is satisfied by its output allocation.

\begin{definition}[$k$-\DominationCondition\ allocation for a \textsc{Line} graph] \label{defn:kfair_lines} Consider a cake-division instance with $n$ agents on a \textsc{Line} and a complete allocation $\mathcal{B} := 
\{B_1, B_2, \cdots, B_n \}$ of some piece $R \subseteq  [0,1]$ of the cake. We say  $\mathcal{B}$ is  $k$-\DominationCondition\  (for $k\in [n]$) if and only if
\begin{enumerate}
     \item [C1.] For $i \geq k$, agent $a_i$  does not envy her neighbors i.e., $v_i(B_i) \geq v_i(B_{i-1})$ and $v_i(B_i) \geq v_i(B_{i+1})$.
     \item [C2.] For agent $a_k$, we have
 $v_{k}(B_{k})=v_{k}(B_{\ell})$ for  all $\ell \le k$.
      \item [C3.] For agent $a_{k+1}$, we have $v_{k+1}(B_{k+1})\ge v_{k+1}(B_\ell)$ for all $\ell \le k$.
 \end{enumerate}
\end{definition}
Note that $1$-\DominationCondition\ allocation are locally envy-free for agents on a \textsc{Line}. And, in particular, we will prove that our protocol $\textsc{Domination}([0,1],1)$ outputs an allocation that is $1$-\DominationCondition, see Lemma~\ref{lem:klogk complexity}.
\paragraph{Overview of \DominationProtocol \ for \textsc{Line} graph:} 
Let us now give a high-level overview of  our recursive protocol, \DominationProtocol\ for some piece $R \subseteq [0,1]$ and some agent $a_k$ for $k \in [n]$.  For any piece $R \subseteq [0,1]$ and any agent $a_k$ for $k \in [n]$, we construct a recursive protocol \DominationProtocol\ that repeatedly invokes $\textsc{Domination}(R,k+1)$ to compute a $k$-\DominationCondition\ allocation of $R$. We call $R$ as residue that keeps decreasing throughout the execution. The protocol consists of a while-loop that terminates when agent $a_{k+1}$ achieves a certain \emph{domination condition} on her value of the (current) residue (as stated in Step 9).

For $k=n$, the protocol $\textsc{Domination}(R,n)$ is defined in a  straight-forward manner: it simply asks agent $a_n$  to cut $R$ into $n$ equal pieces according to her. It can be easily verified that this allocation is indeed $n$-\DominationCondition.  
 For $k<n$, we develop the protocol \DominationProtocol\ that successively constructs a $k$-\DominationCondition \ allocation $\mathcal{A}=\{A_1, \dots, A_n\}$ of piece $R$ among $n$ agents in multiple rounds. It does so by invoking $\textsc{Domination}(R,k+1)$ and using its $(k+1)$-\DominationCondition \ output allocations. Here, $R$ is referred to as the  \emph{residue} and the protocol keeps modifying it throughout its execution until it becomes empty. \DominationProtocol \ consists of a while-loop that runs till $a_{k+1}$ achieves a certain \emph{domination condition} over her value of the current residue (see Step~10). We prove that this is achieved in polynomial many rounds of the while-loop (see Lemma~\ref{lem:klogk complexity}). 
 
 Any round $t$ of the while-loop (Step 3) begins with invoking the protocol $\textsc{Domination}(R^t, k+1)$ to obtain a $(k+1)$-\DominationCondition\  allocation $\mathcal{B}^t=\{B^t_1, B^t_2, \dots, B^t_n\}$ of the current residue $R^t$ (where, $R^{1}=R$). We allocate piece $B^t_j$ to agent $a_j$'s bundle $A_j$ for $j \geq k+2$. Next, agent $a_k$ selects her top $k$ favorite pieces from the remaining $k+1$ pieces $B^t_1, \dots, B^t_{k+1}$ (in Step 8) and the last remaining piece is added to the bundle  $A_{k+1}$ of agent $a_{k+1}$. Let us denote the set of $k$ pieces picked by agent $a_k$ by $S^t$. We remark here that any agent lying to the right of $a_k$ receives an unspoiled\footnote{Without any modification to the bundles obtained from $\textsc{Domination}(R^t,k+1)$.} pieces of cake as obtained from $\textsc{Domination}(R^t,k+1)$ in any round of the while loop. We will use this fact to establish that the final output allocation $\mathcal{A}$ satisfies Condition C1 of $k$-\DominationCondition ness, see Lemma~\ref{lem:klogk complexity}.
 
 In the remaining part of the while-loop, if the \emph{domination condition} for agent $a_{k+1}$ is not yet satisfied, agent $a_k$ will modify the selected $k$ pieces in the set $S^t$ via a \textsc{Trim} procedure (in Step 12). This step dictates agent $a_k$ to trim all the pieces in $S^t$ so that they are all valued equal to her least valued piece in $S^t$. The residue obtained in this procedure becomes the \emph{residue} $R^{t+1}$ for the next round of the while-loop. 
 The least valued piece according to agent $a_{k+1}$ in the output partition of the \textsc{Trim} procedure is added to some appropriately chosen agent $a_w$ for $w \leq k$. The remaining pieces are arbitrarily allocated, one each, to the agents $a_j$ for $j \in [k] \setminus \{w\}$. In Claim~\ref{claim: residue decreasing}, we prove that the above allocation ensures that agent $a_{k+1}$ will eventually achieve the \emph{domination condition} over her value of the (current) residue, see Step 10. The \emph{domination condition} entails that agent $a_{k+1}$ will not envy any agent $a_\ell$ for $\ell \leq k$ even if the entirety of the current residue is allocated to $a_\ell$ in the final allocation. This dominance proves to be crucial in establishing the Condition C3 of $k$-\DominationCondition ness for  the final allocation $\mathcal{A}$. Moreover, note that any agent $a_\ell$ for $\ell \leq k$ receives a bundle from the output partition of \textsc{Trim} and \textsc{Equal} procedures performed by agent $a_k$. Using this fact, we can easily establish Condition C2 of $k$-\DominationCondition ness for the final allocation $\mathcal{A}$, see Lemma~\ref{lem:klogk complexity}.
 
  Finally, to analyze the run-time of our protocol, the key observation is that value of the (current) residue $R$ according to agent $a_{k+1}$ decreases exponentially fast (see Claim~\ref{claim: monotonically increasing}). This ensures that the said domination condition for agent $a_{k+1}$ is achieved in at most $k+k\log k$ rounds of the while-loop (see Claim~\ref{claim: residue decreasing}). After which, the while-loop terminates and agent $a_{k}$ now performs an \textsc{Equal} procedure (instead of \textsc{Trim}) on the output allocation of $\textsc{Domination}(R^{k+k \log k},k+1)$ and hence producing no residue. The $k$ pieces returned by the \textsc{Equal} procedure are arbitrarily assigned to agents $a_1, \dots a_k$, creating the final allocation $\mathcal{A}$.
 It is relevant to note that the only agent who cuts\footnote{This is in sharp contrast to the protocol \textsc{Alg2} for local envy-freeness on \textsc{Depth2Tree}, where the root agent cuts the cake into $n$ equal pieces according to her, see Section~\ref{sec:depth2}.} the cake into $n$ equal pieces according to her throughout the execution of \DominationProtocol \ is the rightmost agent $a_n$. 
 
 We now state the main result of this section, and then formally describe our recursive protocol \DominationProtocol\ followed by analysis of its correctness and query complexity.
\nAgentsOnLine*
 
\begin{algorithm}[h]
 \renewcommand{\thealgocf}{}
    \DontPrintSemicolon
     \SetAlgorithmName{$\textsc{Recursive step}$}{ }{ }
    \caption{Protocol \DominationProtocol\ for $n$ agents on a \textsc{Line}}
        \label{alg:n-line}
       \nonl \textbf{Input:} 
        A cake-division instance $\mathcal{I}$ on a \textsc{Line}, a piece $R \subseteq [0,1]$, and an agent $a_k$ for $k \leq n$. \;
       \nonl \textbf{Output:} A $k$-\DominationCondition\ allocation of $R$.\;
        { \bf Initialize:}  Bundles $A_i \leftarrow \emptyset$ for $i \in [n]$, counter $c =0$\;
        \While{$R\neq\emptyset$}
        {
          $\mathcal{B} \gets \textsc{Domination}(R, k+1)$\;
            \For{$j\geq k+2$}
            {
                $A_j\leftarrow A_j\cup B_j$
            }
         \nonl  --------$\texttt{ Selection}$--------\\
         
           $\mathcal{X}\leftarrow \{B_1,\dots,B_{k+1}\}$\;
            Set $(\mathcal{X}^{(k)},\mathcal{X}) \leftarrow$  \textsc{Select}($a_k,\mathcal{X},k$)\;
            
          $A_{k+1}\leftarrow A_{k+1}\cup  \mathcal{X}$\;
         

        \uIf{$\exists \ i \in [k]$ such that $v_{k+1}(A_{k+1})-v_{k+1}(A_i)\le  v_{k+1}(R)$\label{inalg-domination}}
        {
         \nonl 
 --------$\texttt{Trimming}$--------- \\
       Set $R \gets \emptyset$\;
            $(\mathcal{X}^{(k)},R) \gets \textsc{Trim}(a_k,\mathcal{X}^{(k)})$\;
         Let $X_\ell = \argmin_{1\le i\le k}v_{k+1}(X_i)$ and $w = c \mod{k} + 1$\;
             $A_w\leftarrow A_w\cup X_\ell$ and $\mathcal{X}^{(k)}\leftarrow \mathcal{X}^{(k)}\setminus X_\ell$
             \label{inalg-cupmax} \;
            For each $i \in [k]\setminus \{w\}$, add one arbitrary piece from $\mathcal{X}^{(k)}$  to $A_i$ \;
            $c \rightarrow c+1$ \;
         }                  
         \Else{
         
                  \nonl
 --------$\texttt{Equaling}$--------- \\
         $(\mathcal{X}^{(k)}, X^*) \gets \textsc{Equal}(a_k,\mathcal{X}^{(k)})$\;
           For each $1\le i\le k$, add one arbitrary piece from $\mathcal{X}^{(k)}$  to the bundle $A_i$ \;
         }
      }
        \Return  The allocation $\{A_i\}_{1\le i \le n}$
        \;
\end{algorithm}

 During the execution of \DominationProtocol,  we write $R^t$ to denote the residue at the beginning of  round $t$ of the while-loop. Let us denote the allocation returned by $\textsc{Domination}(R^t,k+1)$ as $\mathcal{B}^t =\{B^t_1, B^t_2, \dots, B^t_n\}$. Furthermore, we write $X^t_1, X^t_2, \dots, X^t_k$ to denote  $k$ equally-valued pieces according to agent $a_k$ obtained after the \textsc{Trim} procedure in Step 11, and $X^t_{k+1} \in \mathcal{B}^t$ be the piece that is added to agent $a_{k+1}$'s bundle, $A_{k+1}$, in Step $9$ of the protocol.  With these notations at our disposal, we begin with the following claim that establishes the fact that the value of the (current) residue keeps reducing for agent $a_{k+1}$ with every iteration of the while-loop.
 
 \begin{claim}
 Consider any round $t$ of the while-loop during the execution of \DominationProtocol, and define $c_t := \max_{\ell \leq k} \{v_{k+1}(X_{k+1}^t) - v_{k+1} (X_{\ell}^{t})\}$ to be the maximum trimmed value according to agent $k+1$ in the \textsc{Trim} procedure. Then, after $k \log k$ rounds of the while-loop, we obtain
\begin{align*}
 v_{k+1}(R^{t + k\log k}) \leq c_t.
\end{align*}
\label{claim: residue decreasing}
\end{claim}
\begin{proof}
Consider any round $t$ of the while-loop of \DominationProtocol. Let us write $\mathcal{B}^t =\{B^t_1, \dots, B^t_n\}$ to denote the $(k+1)$-\DominationCondition\ allocation returned during round~$t$ by $\textsc{Domination}(R^t,k+1)$, where $R^t$ is the residue at the beginning of round $t$. Note that, Condition C3 of $(k+1)$-\DominationCondition ness for $\mathcal{B}^t$ ensures that $a_{k+1}$ values the first $k+1$ bundles in $\mathcal{B}^t$ equally, i.e., $v_{k+1}(B^t_{k+1}) = v_{k+1}(B^t_{\ell})$ for all $\ell \leq k+1$. Since piece $X^t_{k+1}$ that is added to agent $k+1$'s bundle in Step 9 is from one of the first $k+1$ bundles in $\mathcal{B}^t$, we can write 
\begin{align} \label{eq:Bt}
    k\cdot  v_{k+1}(X_{k+1}^t)=  v_{k+1}(\cup_{\ell = 1}^{k} B_{\ell}^t) 
\end{align}
The above equality uses additivity of valuation $v_{k+1}$. Now, the definition of $c_t$ ensures that
\begin{align} \label{eq:ct}
    k\cdot  v_{k+1}(X_{k+1}^t) - \sum_{\ell \leq k} v_{k+1}  (X_{\ell}^{t}) \leq k \cdot c_t
\end{align}
Using Equation (\ref{eq:Bt}) and additivity of $v_{k+1}$, we have
\begin{align*}
    k\cdot  v_{k+1}(X_{k+1}^t) - \sum_{\ell \leq k} v_{k+1}  (X_{\ell}^{t}) &= 
     v_{k+1}(\cup_{\ell = 1}^{k} B_{\ell}^t) - \sum_{\ell \leq k} v_{k+1}  (X_{\ell}^{t})\\
      &=v_{k+1}(\cup_{\ell = 1}^{k} B_{\ell}^t) - v_{k+1}(\cup_{\ell = 1}^{k} X_{\ell}^t) = v_{k+1}(R^{t+1})
\end{align*}
     Here, the last inequality follows from the fact that the residue for the next round comes only from the trimmings. Therefore, combining it with Equation (\ref{eq:ct}), we obtain
     \begin{align} \label{eq:rt+1}
        v_{k+1}(R^{t +1}) \leq  k \cdot c_t
     \end{align}
Next, we will show that agent $a_{k+1}$ values residue $R^{t+1}$ of the next round at most $(1-1/k)$ times that of her value of the current round $R^t$. To see this, observe that agent $a_{k+1}$ values the first $k+1$ bundles in $\mathcal{B}^t$ equally, and out of these $k+1$ bundles, agent $a_k$ trims at most $k-1$ bundles in Step 11 to obtain the residue $R^{t+1}$. Therefore, we can write
\begin{align} \label{eq:bound}
    v_{k+1}(R^{t+1}) &\leq (k-1) v_{k+1}(X^t_{k+1})  \leq (1-1/k) \cdot v_{k+1}(R^t)
\end{align}
For the last inequality, let us assume otherwise, i.e.,  we have $ v_{k+1}(X^t_{k+1}) >(1/k) \cdot  v_{k+1}(R^t)$. We can then write $\sum_{\ell=1}^k v_{k+1}(X^t_{\ell}) = k.v_{k+1}(X^t_{k+1})> v_{k+1}(R^t)$, leading to a contradiction. Using Equations~(\ref{eq:rt+1}) and (\ref{eq:bound}), we can now bound the size of the residue after $k\log k$ rounds of the while-loop,
\begin{align*}
    v_{k+1}(R^{t+ k\log k}) &\leq \left(1 - \frac{1}{k} \right)^{k\log k} \cdot v_{k+1}(R^{t+1}) \\
    &\leq \left(1 - \frac{1}{k}\right)^{k\log k} \cdot k \cdot c_t  \leq c_t 
\end{align*}
 This completes our proof.
\end{proof}

We now establish a useful relation between agent $a_{k+1}$'s value for her own bundle and the bundle of agent $a_j$ for $j \leq k$, for two consecutive rounds of the while-loop in \DominationProtocol.
\begin{claim}
During the execution of \DominationProtocol, the cumulative difference between the value of agent $a_{k+1}$ for her own bundle and for any of the first $k$ bundles increases monotonically with each round of the while-loop i.e., for any round $t$, we have
$$ v_{k+1}(A_{k+1}^{t}) - v_{k+1}(A_j^{t}) \leq v_{k+1}(A_{k+1}^{t+1}) - v_{k+1} (A_j^{t+1}) \quad \text{for all} \ j \leq k$$
\label{claim: monotonically increasing}
\end{claim}

\begin{proof}
Consider any round $t$ of the while-loop of \DominationProtocol. Write $\mathcal{B}^t =\{B^t_1, B^t_2, \dots, B^t_n\}$ to denote the allocation returned by $\textsc{Domination}(R^t,k+1)$ in round $t$. Furthermore, we write $\mathcal{A}^t$ and $\mathcal{A}^{t+1}$ to denote the  allocations at the end of rounds $t$ and $t+1$ of the while-loop.
Note that, agent $a_{k+1}$ gets an unspoiled piece from $\mathcal{B}^{t}$ in Step~9, i.e., we have  $X_{k+1}^t = A_{k+1}^{t+1} \setminus A_{k+1}^t $. Also for all $j \leq k$, the pieces $B_j^t$ are trimmed in Step~12 to obtain $X_j^{t}$ that is, $X_j^{t}\subseteq B_j^t$.  Furthermore, since $\mathcal{B}^t$ is $(k+1)$-\DominationCondition, Condition C3 ensures that $v_{k+1}(X_{k+1}^{t}) = v_{k+1}(B_j^t) $ for all $j \leq k$. Therefore, for all  $j \leq k$ we have,   
\begin{align*}
    v_{k+1}(X_{k+1}^t) = v_{k+1}(A_{k+1}^{t+1}) - v_{k+1} (A_{k+1}^t) = v_{k+1}(B_j^t) \geq  v_{k+1}(X_j^t) 
\end{align*}
 Assuming that the piece $X_j^t$ is added to some bundle $A_p \in \mathcal{A}^t$ for $p \leq k$ in Steps 14 and 15, we obtain
\begin{align*}
     v_{k+1}(A_{k+1}^{t+1}) - v_{k+1} (A_{k+1}^{t}) \geq  v_{k+1}(X_j^{t}) =  v_{k+1}(A_p^{t+1}) - v_{k+1} (A_p^{t}).
\end{align*}
This establishes the stated claim.
\end{proof}

Next, we show an important lemma that formally establishes the crux of our recursive protocol: given access $(k+1)$-\DominationCondition\ allocations of $R \subseteq [0,1]$, \DominationProtocol\ computes a $k$-\DominationCondition\ allocation in at most $k+k\log k$ rounds of the while-loop.

\begin{lemma}\label{lem:klogk complexity}
Given any cake-division instance with $n$ agents on a \textsc{Line}, consider the execution of the protocol \DominationProtocol \ for any piece $R \subseteq [0,1]$ and any agent $a_k$ for $k \in [n]$. Then, the protocol returns a $k$-\textsc{Fair} allocation  in $k+ k\log k$ rounds of the while-loop. 
\end{lemma}
\begin{proof}
We begin by proving that \DominationProtocol\ indeed returns a $k$-\DominationCondition \ allocation, and then we will establish the desired count on the number of while-loops required to achieve so.\\

\noindent
{\em Correctness:} For a given piece $R \subseteq [0,1]$ and $k \in [n]$, we will prove that the output allocation $\mathcal{A}$ of \DominationProtocol\  is $k$-\DominationCondition. We will prove it via induction on the position of the agent $a_k$ on the \textsc{Line}. Recall that $\textsc{Domination}(R,n)$ asks agent $a_n$ to simply divide $R$ into $n$ equal pieces, and hence the required condition trivially follows for $a_n$. Now, let us assume that the claim holds true for $k+1$, and we will prove it for $k$. We will prove that \DominationProtocol \ is $k$-fair by proving that is satisfies the Condition C2, then C3 and then finally C1 (for the ease of presentation).

\begin{itemize}[leftmargin=*]
   \item { \bf \underline{Condition C2}}:
    Note that, in any round $t$, agent $a_k$ selects her $k$ most preferred pieces from $k+1$ pieces $B^t_1, \dots, B^t_{k+1} \in \mathcal{B}^t$ in Step~8, and trims all of them to have value equal to the lowest valued piece. Hence, in each round of the while-loop, we have $v_k(A_k^t) = v_k(A_{i}^{t})$ for $i \leq k$. The \texttt{Equaling} phase maintains this property, hence establishing Condition C2 of k-\DominationCondition ness.

    \item {\bf \underline{Condition C3}}: Recall that in any round $t$, the first $k$ agents receives a bundle that is a subset of the first $k+1$ bundles of the $(k+1)$-\textsc{Fair} allocation returned by $\textsc{Domination}(R^t,k+1)$. Therefore, we have $v_{k+1}(A^t_{k+1}) \geq v_{k+1}(A^t_j)$ for all $j \leq k$. Finally, at the termination round $T$ of the while-loop, we have   $$v_{k+1} (A^T_{k+1}) - v_{k+1} (A^T_j) \geq v_{k+1}(R^T) \  \text{for all} \ j\leq k.$$ That is, the residue $R^T$ is small enough that it does not induce any envy for $a_{k+1}$ even if it is fully allocated  to any  bundle $A^T_j$ for $j \leq k$. Recall that in the last round, agent $a_k$ performs the \textsc{Equal} procedure on some $k$ pieces of the allocation obtained by $\textsc{Domination}(R^T,k+1)$. Since the whole of the residue $R^T$ induces no envy for agent $a_{k+1}$, even after the \texttt{Equaling} phase, we maintain that $v_{k+1}(A_{k+1}) \geq v_{k+1}(A_j)$ for all $j \leq k$. Hence, Condition C3 is satisfied.
    
    \item { \bf \underline{Condition C1}}: 
    For any round $t$ of the while-loop, by induction hypothesis, we know that the output allocation $\mathcal{B}^t$ of $\textsc{Domination}(R^t,k+1)$ is $(k+1)$-\textsc{Fair}. That is, any agent $a_j$ with $j \geq k+1$ does not envy her neighbors in $\mathcal{B}^t$. Since every agent $a_j$ with $j \geq k+2$ is allocated unspoiled bundles from $\mathcal{B}^t$ in Step~5, we obtain that these agents (i.e. $a_j$ with $j \geq k+2$) do not envy their neighbors in the output allocation of \DominationProtocol. Due to the same reason, $a_{k+1}$ does not envy her neighbor $a_{k+2}$. To see why agent $a_{k+1}$ does not envy her other neighbour, $a_k$, refer to the Condition C3 above. 

    Finally, we focus on agent $a_k$. In any round $t$, she selects her $k$ most preferred pieces from $k+1$ pieces $B^t_1, \dots, B^t_{k+1} \in \mathcal{B}^t$ in Step~8, and the remaining piece is allocated to agent $a_{k+1}$. The \texttt{Trimming} and \texttt{Equaling} phase then ensures that we obtain $v_k(A_k) \geq v_k(A_{k+1})$. The fact that $a_k$ does not envy her neighbor, $a_{k-1}$, follows from Condition C2, proved above. In conclusion, Condition C1 is satisfied.
\end{itemize}

\noindent
{\em Run-time:}
Consider the first $k$ rounds of the while-loop during the execution of \DominationProtocol. If the algorithm terminates before $k$ rounds, we are done. If not, observe that Step~12 allocates the smallest piece after trimming (i.e. $\arg\min_{1 \leq i \leq k}v_{k+1}(X_i)$) to  different bundles in the first $k$ rounds.  Assume, without loss of generality, bundle $A_{\ell}$ is allocated the smallest trimmed piece in round $\ell$. For all $\ell \leq k$, we can write   
\begin{align*}
    v_{k+1}(A_{k+1}^{k+ k\log k}) -  v_{k+1}(A_{ \ell }^{k+ k\log k}) & \geq  v_{k+1}(A_{k+1}^{\ell+1}) -  v_{k+1}(A_{\ell}^{\ell+1}) \tag{by Claim \ref{claim: monotonically increasing}}\\
    & \geq c_\ell \tag{by Step 12}\\
    &\geq v_{k+1}(R^{k + k\log k}) \tag{by  Claim \ref{claim: residue decreasing}}
 \end{align*}
Therefore, after at most $k+ k\log k$ rounds, the while-loop terminates and the algorithm enters the \texttt{Equaling} phase to output the final $k$-\DominationCondition \ allocation.
\end{proof}

\paragraph{Proof of Theorem~\ref{thm:line}:}
Let us denote the query complexity of $\textsc{Domination}(R,k)$ by $T_k$ for $k \in [n]$. Note that, each round of the while-loop and the \texttt{Equaling} phase during the execution of the protocol \DominationProtocol\ requires $k$ \texttt{eval} queries and $k-1$ \texttt{cut} queries. Lemma~\ref{lemma:domination} proves that \DominationProtocol\ outputs a $k$-\DominationCondition\ allocation in $k+k \log k$ rounds of the while-loop. We will prove that $T_k\le \sum_{j=k}^n j \prod_{i=k}^j (3i\log i)$ using induction on $k$. 
For the base case $k=n$, we know that $T_n = n \le 3 n^2 \log(n)$. Assuming that the bound holds true for $T_{k+1}$, we begin by writing the recursion formula that relates $T_k$ and $T_{k+1}$,
\begin{align*}
    T_k&=(k+k\log k)T_{k+1}+k^2\log k\\
    &\le 3k\log k\sum_{j=k+1}^n j \prod_{i=k+1}^j (3i\log i)+k^2\log k\\
    &\le \sum_{j=k}^n j \prod_{i=k}^j (3i\log i)
\end{align*}
Now, to bound the run-time of the \DominationProtocol\, we need to bound $T_2$.
With above bound at hand, we get  $T_2\le \sum_{j=2}^n j \prod_{i=2}^j (3i\log i)$. Let us now denote $c_j=j\prod_{i=2}^j (3i\log(i))=j\cdot 3^j\cdot j!\prod_{i=2}^j \log i$. Note that, for all  $j<n$ we have $ 2c_j<c_{j+1}$, and hence, 
$\sum_{j=2}^{n-1}c_j<c_n$. We can therefore write
\begin{align*}
    T_2 &\leq \sum_{j=2}^n j \prod_{i=2}^j (3i\log i) \\
    & \leq \sum_{j=2}^n c_j \tag{by definition of $c_j$}\\
    &\le 2c_n=2n\cdot 3^n \cdot n!(\log n)^n
\end{align*}
 Using Stirling's approximation, we obtain  $T_2=n^{(1+o(1))n}$. This establishes the stated claim and completes the proof.
  \qed

\section{Local Envy-freeness on a \textsc{Tree}}
\label{sec:tree}

In this section, we prove our main result (Theorem~\ref{thm:line}) that establishes an upper bound of $n^{(1+o(1))n}$ on the query complexity for finding locally envy-free allocations among $n$ agents on a \textsc{Tree}.

Without loss of generality, we assume that the underlying tree graph $G$ is rooted at agent $a_n$ and agents ${a_1,a_2,\ldots a_n}$ are indexed according to some (arbitrary) topological sorting. Note that, agent $a_1$ is therefore a leaf in $G$. Similar to Section~\ref{sec:line}, we primarily construct a recursive protocol $\textsc{Domination}([0,1],1)$ for a \textsc{Tree} using the techniques developed for \textsc{Line} graphs. We will generalize our previous protocol \DominationProtocol\ for \textsc{Line} graphs and carefully adapt the key notion of $k$-\DominationCondition\ allocations for \textsc{Tree} graphs. Interestingly, we will see that during this generalization, we are able to maintain an identical query complexity bound and incur no loss whatsoever.

To formally explain the required modifications in the notion of $k$-\DominationCondition ness for \textsc{Tree} graph, we would need to expand our terminology, and hence we begin with the following useful definitions.\\

 \noindent
 \textbf{Terminology:} 
First, note that the topological ordering over agents ensures that any descendant of agent $a_j$ for $ j\in [n]$ must have an index that is smaller than $j$. For any agent $a_j$ for $j \in [n]$, we begin with the following notations.
\begin{itemize}
    \item $D_j :=   \{j\} \cup \{i \in [n]: \text{$a_i$ is a descendant of $a_j$}\} $ is a collection of the index $j$ and the indices of agents who are descendants of $a_j$. In particular, we will write $d_j = |D_j|$ to denote the size of the set $D_j$ for agent $a_j$.
    \item $C_j$ denotes the set of (immediate) children of $a_j$, and we write $p_j$ to denote the index of the parent of agent $a_j$. 
\end{itemize}    
Next, we define the following useful notions for any agent $a_j$ with respect to a fixed index $k \in [n]$.
\begin{itemize}
    \item An agent $a_j$ is said to be \emph{active} if $j >  k$, otherwise she is said to be \emph{inactive}.
     \item $\texttt{Inchild}(k,a_j):=\{i\mid i\le k\text{ and } i\in C_j\}$ is a collection of the indices of inactive children of agent $a_j$ with respect to the index $k$. Since any agent $a_j$ is active for $j>k$, the set $\texttt{Inchild}(k,a_j) = \emptyset$ for all agents $a_j$ with $j>k$.
    \item $\texttt{Inact}(k,a_j):=  \{j\} \cup \{t\in D_i\mid i\in \texttt{Inchild}(k,a_j)\}$ is a collection of the index $j$ and the indices of descendants of inactive children of $a_j$ (with respect to $k$). 
    \item $\Storage(k,a_j):=\{\cup_i B_i \mid i\in \texttt{Inact}(k,a_j)\}$ represents the bundles assigned to agents with indices in $\texttt{Inact}(k,a_j)$ in an allocation $\mathcal{B}=\{B_1, \dots, B_n\}$.
\end{itemize}
 For an allocation $\mathcal{B}$, the collection of storage sets  $\{\Storage(k,a_j)$ for $j \in [n]$\} creates a partition of its bundles.  Note that, $\Storage(k,a_j) = \emptyset$ for all agents $a_j$ with $j>k$ since the corresponding set $\texttt{Inchild}(k,a_j) = \emptyset$. When the value of $k$ is obvious from the context, we will use the phrase "the storage of $a_j$" to mean $\texttt{Storage}(k,a_j)$ set.
\begin{figure}
    \centering
    \includegraphics[scale=0.6]{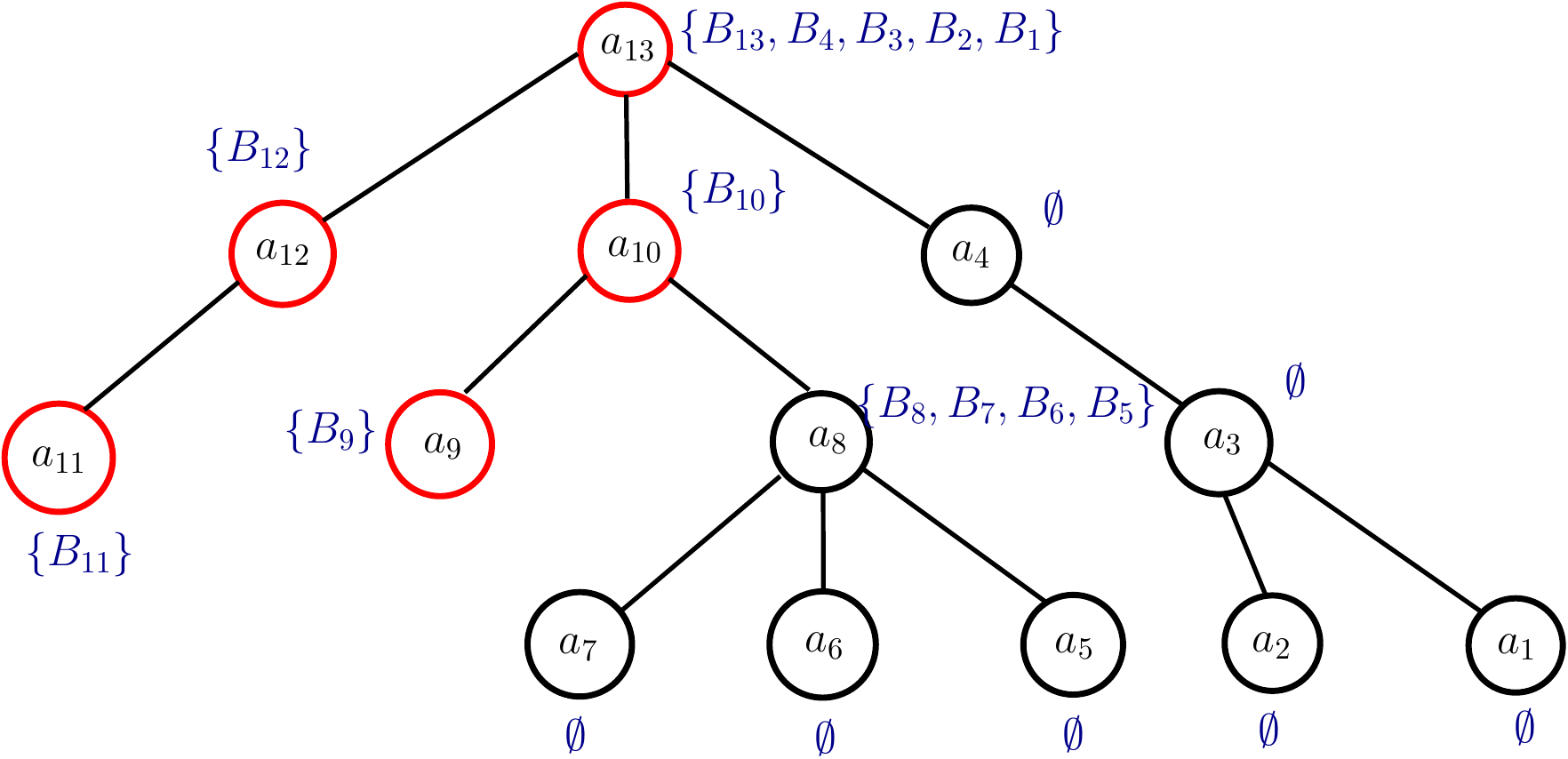}
    \caption{A representative example of graphical constraints with 13 agents on a {\sc Tree}. The identity of a node is written inside the circle. For a fixed index $k=8$, the black nodes represent the set of inactive agents while the red nodes  correspond to the active agents. The blue set written next to an agent  $a_j$ represents the set $\texttt{Storage}(k, a_j)$ for an allocation~$\mathcal{B}$.}
    \label{fig:my_label}
\end{figure}

Equipped with the above notions, we are ready to extend the ideas and techniques developed in Section~\ref{sec:line} for computing locally envy-free allocation among agents on a \textsc{Line} to that on a \textsc{Tree}.  For doing so, we generalize the notion of $k$-\DominationCondition ness for \textsc{Tree} graphs appropriately, which will provably form the crux of our techniques and the protocol \DominationProtocol\ for \textsc{Tree} graphs.

\begin{definition}[$k$-\DominationCondition\ allocation for a \textsc{Tree} graph] \label{defn:kfair_trees} Consider a cake-division instance with $n$ agents on a \textsc{Tree} (with topological ordering on the agents with $a_n$ being the root agent). For a given piece  $R \subseteq  [0,1]$, we say an allocation $\mathcal{B}$ (of $R$) is $k$-\DominationCondition\ if for each agent $a_j$ with $j\ge k$ the following conditions hold.
 \begin{enumerate}
    \item[C1.] Agent $a_j$ does not envy her neighbours.
    \item[C2.] $v_j(B_j)=v_j(B) $ for all $ B \in \texttt{Storage}(k-1,a_j)$,  and 
    \item[C3.] $v_j(B_j)\ge v_j(B)$  for all $B \in \texttt{Storage}(k-1,a_\ell)$ such that $a_\ell$ is an active child of $a_j$.
\end{enumerate}
\end{definition}
Let us now understand how does the above-defined $k$-\DominationCondition\ property for \textsc{Tree} graph is an extension of its counterpart in the {\sc Line} graph (as stated in Definition~\ref{defn:kfair_lines}). To begin with, Condition C1 is identical in both cases. Next, observe that for a {\sc Line} graph, we have
\begin{equation*}
\Storage(k-1,a_j) = \begin{cases}
  \{B_1, \dots, B_k\} &\text{if}\ j=k  \\
 \{B_j\} &\text{if}\ j \geq k+1 \\
 \emptyset &\text{if}\ j<k
\end{cases}
\end{equation*}

Condition C2 requires that every agent with index higher than $k$ values all the bundles in her storage set (with respect to $k-1$) equally. Since each agent $a_k$ for $j \geq k+1$ maintains a single bundle in her storage set in {\sc Line} graph, it is sufficient to compare this condition only for agent $a_k$. And, since $\Storage(k-1,a_k)=\{B_1, \dots, B_k\}$ for \textsc{Line} graph, Condition C2 for \textsc{Tree} graphs when computed for \textsc{Line} graphs coincides with Condition C2 in Definition~\ref{defn:kfair_lines}. 

Finally, Condition C3 says that any agent with index higher than $k$ should (weakly) prefer her own bundle over that of any bundle in the storage set of any of her active child (with respect to index $k-1$). In the case of {\sc Line} graph, $a_k$ has no active children, while for  $j \geq k+1$, agent  $a_{j-1}$ is the only active child of $a_j$. As stated earlier, $\Storage(k-1,a_k)=\{B_1, \dots, B_k\}$ and hence, Condition C3 dictates that $v_{k+1}(B_{k+1}) \geq v_{k+1}(B)$ for $B \in \Storage(k-1,a_k)=\{B_1, \dots, B_k\}$. This again matches exactly with Condition C3 for \textsc{Line} graphs. Therefore, Definition~\ref{defn:kfair_trees} rightfully extends the notion of $k$-\DominationCondition ness for \textsc{Line} graphs (as stated in Definition~\ref{defn:kfair_lines}) to \textsc{Tree} graphs.

We now provide an example to illustrate the conditions for $k$-\DominationCondition ness in \textsc{Trees} for better understanding.

\begin{example}
Consider the example of an underlying social graph of a \textsc{Tree} with $13$ agents given in Figure~\ref{fig:my_label}. For a fixed index $k=8$, a $k$-\DominationCondition \ allocation $\mathcal{B} = \{B_1, \dots, B_{13}\}$ must satisfy the following conditions. 
\begin{itemize}
    \item Condition C1: Agents $a_9, a_{10}, \dots, a_{13}$ do not envy their neighbors. 
    \item Condition C2: Note that, the storage sets of agents are as following
    \begin{equation*}
\Storage(7,a_j) = \begin{cases}
  \{B_5, \dots, B_8\}  &\text{if}\ j=k  \\
 \{B_j\} &\text{if}\ j \geq k+1 \\
 \emptyset &\text{if}\ j<k
\end{cases}
\end{equation*}
    Hence, we must have
     $v_8(B_8) = v_8(B)$ for all $B \in \{B_5,B_6, B_7, B_8 \}$.
    \item Condition C3: Agent $a_{10}$ has two active children (with respect to index $7$): $a_9$ and $a_8$ with storage sets $\{B_9\}$ and $\{B_5, \dots, B_8\}$ respectively. Hence, this condition dictates that we have $v_{10}(B_{10}) \geq v_{10}(B)$ for $B \in \{B_5, \dots, B_9\}$.
    We can similarly state the guarantees for the remaining agents $a_j$ for $j \geq k$.
\end{itemize}
\end{example}

We finally state and prove our main result that establishes an upper bound of $n^{(1+o(1))n}$ on the query complexity for local envy-freeness on \textsc{Tree} graphs. Interestingly, it enjoys similar algorithmic guarantees as that of its counterpart in the \textsc{Line} graph (see Theorem~\ref{thm:line}). Designing a discrete and bounded protocol for local envy-freeness on trees is listed as an open problem in \cite{AbebeKP16}. We answer this open problem by proving the following result.
\nAgentsOnTree*
We begin with detailing our recursive protocol $\textsc{Domination([0,1],1)}$ for \textsc{Tree}. We describe the necessary modifications to its counterpart for \textsc{Line} with the new terminology that generalizes its key ingredients. As before, for any piece $R \subseteq [0,1]$ of the cake and any agent $a_k$ for $k \in [n]$, we construct a recursive protocol \DominationProtocol\ that repeatedly invokes $\textsc{Domination}(R,k+1)$ to compute a $k$-\DominationCondition\ allocation of $R$. We call $R$ as residue that keeps decreasing throughout the execution. The protocol consists of a while-loop that terminates when the parent agent $a_{p_k}$ of agent $a_k$ achieves a certain \emph{domination condition}\footnote{Recall that \DominationProtocol\ for \textsc{Line} graph requires agent $a_{k+1}$ to achieve the domination condition. If we consider \textsc{Line} as a tree rooted at $a_n$ then $a_{k+1}$ is indeed the parent of agent $a_k$.} on her value of the (current) residue (as stated in Step 9).

Recall that the set $D_k$ consists of the indices of the descendants of agent $a_k$ and her own index, and we write $d_k = |D_k|$ to denote its cardinality. As before, any round $t$ of the while-loop begins with invoking $\textsc{Domination}(R^t,k+1)$ where $R^t$ is the residue at the beginning of round $t$. We denote $\mathcal{B}^t$ to denote the $(k+1)$-\textsc{Fair} allocation returned by $\textsc{Domination}(R^t,k+1)$. Now, agent $a_k$ selects her $d_k$ favorite pieces from her $\Storage(k,a_{p_k})$ set in $\mathcal{B}^t$ (see Steps 4 and 5). We denote these chosen $d_k$ bundles by set $S^t$ and re-index its bundles so that they bear the indices in $D_k$. If the domination condition for agent $a_{p_k}$ is not yet achieved, agent $a_k$ performs a \textsc{Trim} procedure (Step 11), otherwise it performs an \textsc{Equal} procedure (Step 17).

We will prove that \DominationProtocol\ outputs a $k$-\DominationCondition\ allocation in $d_k+d_k\log d_k$ iterations of the while-loop (Lemma~\ref{lemma:treedomination}). Since the social graph is a \textsc{Tree}, we know that $d_k+d_k\log d_k \leq k \log k$ for any $k \in [n]$. And hence, we obtain an overall query complexity bound of $n^{(1+o(1))}$ similar to that of the case of \textsc{Line}.

\begin{algorithm}[ht!]
 \renewcommand{\thealgocf}{}
    \DontPrintSemicolon
     \SetAlgorithmName{$\textsc{Recursion step}$}{ }{ }
    \caption{\DominationProtocol\ for \textsc{Trees}}
        \label{alg:n-tree}
       \nonl \textbf{Input:} 
        A cake-division instance $\mathcal{I}$ on a \textsc{Tree} rooted at agent $a_n$ and where agents are indexed according to a topological order, a piece $R \subseteq [0,1]$, and an agent $a_k$ for $k \leq n$.\;
       \nonl \textbf{Output:} A $k$-\DominationCondition\ allocation of $R$. \;
        Initialize  $R \leftarrow U$, bundles $A_i \leftarrow \emptyset$, counter $c=0$\;
        \While{$R\neq\emptyset$}
        {
          $\mathcal{B} \gets  \textsc{Domination}(R,k+1)$\;

         \nonl  --------$\texttt{ Selection}$--------\\
           $\mathcal{X}\leftarrow \texttt{Storage}(k,a_{p_k})$\;
            Set $(\mathcal{X}^{(k)},\mathcal{X}) \leftarrow$  \textsc{Select}($a_k,\mathcal{X},|D_k|$)\;
            
             Let us re-index the bundles in $\mathcal{X}^{(k)}$ so that they bear the indices in $D_k$\;
            \tcc{We can do this because Step $3$ ensures that agent $a_{p_k}$ is indifferent towards the bundles in $\Inact(k,a_{p_k})$}
            \For{$j \notin D(k) $}
            {
                $A_j\leftarrow A_j\cup B_j$
            }

        \uIf{$\exists \ i\in D_k$ such that  $v_{p_k}(A_{p_k})-v_{p_k}(A_i) \le  v_{p_k}(R)$\label{inalg-domination}}
        {
          \nonl 
 --------$\texttt{Trimming}$--------- \\
       Set $R \gets \emptyset$\;
        
            $(\mathcal{X}^{(k)},R) \gets \textsc{Trim}(a_k,\mathcal{X}^{(k)})$\;
            Let  $X_t = \argmin_{i\in D(k) }v_{p_k}(X_i)$ and $w = c \mod |D_k| +1 $\; \label{inalg-argmin}
             $A_w\leftarrow A_w\cup X_t$ and $\mathcal{X}^{(k)}\leftarrow \mathcal{X}^{(k)}\setminus X_t$
             \tcc{Trying to achieve domination on $A_w$ for the agent $a_{p_k}$}\label{inalg-cupmax}
            For each $i\in D_k\setminus\{w\}$, add one arbitrary piece from $\mathcal{X}^{(k)}$  to $A_i$ \;
            $c \rightarrow c+1$ \;
         }                  
         \Else{
         
                  \nonl
 --------$\texttt{Equaling}$--------- \\
             $(\mathcal{X}^{(k)}, X^*) \gets \textsc{Equal}(a_k,\mathcal{X}^{(k)})$\;
           For each $i\in D_k$, add one arbitrary piece from $\mathcal{X}^{(k)}$  to $A_i$ \;
        
         }

      }
        \Return  The allocation $\{A_i\}_{1\le i \le n}$\;
\end{algorithm} 

We now state Claim~\ref{claim:tree_residue} that establish the fact that the value of the residue for the parent agent $a_{p_k}$ keeps on decreasing with each iteration of the while-loop in \DominationProtocol and Claim~\ref{claim:tree_difference} that proves a useful relation between the values (according to agent $a_{p_k}$) of $a_{p_k}$'s bundle and any bundle of agents in the set $D_k$ in two consecutive rounds of the while-loop in \DominationProtocol. The proofs of these two claims follow from directly adapting the corresponding proofs of  Claims~\ref{claim: residue decreasing} and \ref{claim: monotonically increasing}, and hence we omit them here.

\begin{claim} \label{claim:tree_residue}
 Consider any round $t$ of the while-loop during the execution of \DominationProtocol, and define $c_t := \max_{\ell \leq d_k} \{v_{p_k}(X_{k+1}^j) - v_{p_k} (X_{\ell}^{j})\}$ to be the maximum trimmed value according to parent agent $a_{p_k}$ in the \textsc{Trim} procedure. Then, after $d_k \log d_k$ rounds of the while-loop, we obtain
$$ v_{p_k}(R^{t + d_k \log d_k}) \leq c_t$$
\end{claim}

\begin{claim} \label{claim:tree_difference}
During the execution of \DominationProtocol, the cumulative difference between the value of parent agent $a_{p_k}$ for her own bundle and for any bundle corresponding to agents in $D_k$, increases monotonically with each round of the while-loop i.e., for any round $t$, we have
$$ v_{p_k}(A_{k+1}^{t}) - v_{p_k}(A_{\ell}^{t}) \leq v_{p_k}(A_{k+1}^{t+1}) - v_{p_k} (A_{\ell}^{t+1}) \quad \text{ for all } \ell \leq d_k$$
\end{claim}

Next, we state and prove the key lemma of this section that proves that our generalization of the recursive protocol \DominationProtocol\ developed in Section~\ref{sec:line} for \textsc{Line} to \textsc{Tree} is query-efficient. In other words, we preserve the corresponding query complexity bounds. Note that, Lemma~\ref{lemma:treedomination} helps us to directly generalize the proof of Theorem~\ref{thm:line} for establishing Theorem~\ref{thm:tree}. Hence, all that is left is to prove Lemma~\ref{lemma:treedomination}.
\begin{lemma} \label{lemma:treedomination}
Given any cake-division instance with $n$ agents on a \textsc{Tree}, consider the execution of the protocol \DominationProtocol \ for any piece $R \subseteq [0,1]$ and any agent $a_k$ for $k \in [n]$. Then, the protocol returns a $k$-\textsc{Fair} allocation  in $d_k+ d_k \log d_k \le k+k\log k$  rounds of the while loop. 
\end{lemma}


\begin{proof} 
Given a tree over $n$ agents with topological ordering, and rooted at $a_n$, we begin by proving that \DominationProtocol\ indeed returns a $k$-\DominationCondition \ allocation, and then we will establish the desired count on the number of while-loops required to achieve so.\\

\noindent
{\em Correctness}: Let us fix some piece $R \subseteq [0,1]$. We will prove that, for any $k \in [n]$, the output allocation, $\mathcal{A}$, of the \DominationProtocol\ is $k$-\DominationCondition \  via induction on the position of the agent $a_k$. Recall that $\textsc{Domination}(R,n)$ asks agent $a_n$ to simply divide $R$ into $n$ equal pieces, and hence the required condition trivially follows for $a_n$. Now, let us assume that the claim holds true for $k+1$, and we will prove it for $k$. Consider any round $t$ of the while-loop during the execution of \DominationProtocol. Let us denote the allocation returned by $\textsc{Domination}(R^t,k+1)$ in this round by $\mathcal{B}^t$. We will prove that \DominationProtocol \ is $k$-\textsc{Fair} by proving that is satisfies all three conditions.

\begin{itemize}[leftmargin=*]
    \item { \bf \underline{Condition C1}:} By induction hypothesis, we know that $\mathcal{B}^{t}$ is  $(k+1)$-\DominationCondition. That is, any agent $a_j$ with $j \geq k+1$ does not envy her neighbors in $\mathcal{B}^t$. Since we have indexed agents according to topological ordering, it follows that any agent $a_j$ with $j \geq k+1$ is allocated an unspoiled bundle from $\mathcal{B}^t$ in Steps 7-8, and hence she does not envy her neighbor who is not in the set $D_k$. We now show that they don't envy their neighbors from $D_k$ either. We do it by proving the following two facts:
    
    \begin{itemize}[leftmargin=*]
            \item [-] \textit{Any agent $a_m$ such that $m \in D_k\setminus \{k\}$ is not a neighbor of $a_j$:}  First note that if  $a_m \in D_k$ and $a_m \neq a_k$, then topological ordering dictates that $m<k$ and $p_m \leq k$.  If $a_m$ is a neighbor of $a_j$ then $a_j$ has to be the parent of $a_m$ because $ j > k > m$.  This is not possible since $p_m \leq k$. 
        
        \item [-] \textit{Agent $a_j$ does not envy agent $a_k$:} The indexing on agents implies that the only agent $a_j$ with $j \ge k+1$ who is a neighbor of agent $a_k$ is $a_{p_k}$. 
        The algorithm terminates when the domination condition in Step $9$ is satisfied for agent $a_{p_k}$.  That is, for the last round $T$ of the while-loop we have $$  v_{p_k } (A_{p_k}^T) \geq v_{p_k} (A_{m}^T) + v_{p_k} (R^T)$$ 
        for all $m \in \Inact(k,a_{p_k})$.
        This implies that, even after the \texttt{Equaling} phase, agent $a_{p_k}$ does not envy any bundle in the storage of $a_k$.
        
    \end{itemize}
    Finally, we need to show that $a_k$ does not envy any of its neighbors in the output allocation as well. It is easy to see that $a_k$ does not envy its parent since $a_k$ selects her preferred pieces in Step~5 before her parent agent. Furthermore, $a_k$ does not envy any of her children  because $a_k$ trims the selected pieces to make them equal to lowest-valued piece. Hence, in each round of the while-loop, $a_k$ values each piece $A_{i}^{t}$ with $i \leq k$ equally. 
    
    \item  { \bf \underline{Condition C2}:}
     We begin with the following observation. 
     \begin{observation}
                The set $\Inact(k-1,a_j) \subseteq   \Inact(k, a_j)$ for every $j \geq k +1$.
                \label{obs: inchild relation}
    \end{observation}
Fix any $j \geq k+1$. Note that, going from $k-1$ to $k$, $a_k$ is the only agent who becomes active from inactive, and therefore we can write
    \begin{equation}
        \Inchild (k,a_j) = \begin{cases}
        \Inchild (k-1,a_j) \cup \{a_k\} & \text{ if } a_k \in C_j\\
         \Inchild (k-1,a_j) & \text{ if } a_k \notin C_j
        \end{cases}
        \label{eq: inchild ralation first case}
    \end{equation}
    That is, we obtain $ \Inchild (k-1,a_j) \subseteq  \Inchild (k,a_j)$.
   \qed

From Observation \ref{obs: inchild relation} and condition C2 for {\sc DominationProtocol }$(R,k+1)$ we have that,  for all $j \geq k+1$, $ v_{j}(B_{j}^{t})  = v_{j}(B^{t}) $ for all $B \in \texttt{\Storage}(k-1, a_j)$.  Note that every agent $j\geq k +1 $ gets an unspoiled piece in Step 8 of the protocol and hence, for any round $t$ of the while-loop, we have  $v_{j}(A_{j}^{t}) =  v_{j}( \cup_{s=1}^t B_{j}^{s}) = v_{j}(\cup_{s=1}^t B_{i}^{s}) = v_{j}(A_{i}^{t}) $ for all $i \geq k+1$.

To conclude,  it is enough to show that $v_k(A_k^{t}) = v_k(A_i^{t}) $ for all $i \in \Inact(k-1, a_k)$. This is true since in any round of the while-loop, these bundles are of same value according to agent $a_k$ (after Steps~11 and 17).  

  \item  { \bf \underline{Condition C3}:}
 We need to show that $v_j(B_j)\ge v_j(B)$  for all $B \in \texttt{Storage}(k-1,a_{\ell})$ such that $a_{\ell}$ is an active child of $a_j$. We consider the following two cases. 
 \begin{itemize}[leftmargin=*]
     \item $\underline{\ell \geq k+1}$:  First, note that the induction hypothesis implies that the output allocation $\mathcal{B}^t$, of the $\textsc{Domination}(R,k+1)$ satisfies condition C3 for each round $t$ of the while-loop. That is, we have  $ v_j(B_j^t) \geq v_j(B)$ for all $B \in \texttt{Storage}(k, a_{\ell})$ for the allocation $\mathcal{B}^t$ obtained in Step~3. Next, from Observation  \ref{obs: inchild relation}, we have $\texttt{Storage}(k-1, a_{\ell}) \subseteq \texttt{Storage}(k, a_{\ell})$. Hence, we obtain that condition C3 is satisfied for $\ell \geq k+1$ till Step 3.
     
     Finally,  note that by topological ordering we know that $j, \ell \notin D_k$ as  $j > \ell \geq k+1$ and  in Step 8 all the agents $a_m \notin D_k$ are assigned unspoiled pieces i.e., $A_m^{t+1} = A_m^{t} \cup B_m^{t}$.  Hence, condition C3 is maintained throughout the execution of the while-loop.    
     \item $\underline{\ell < k+1}$: Recall that  condition C3 considers an active child $a_{\ell}$ of $a_j$. Further, observe that the set  $\texttt{Inact}(k-1, a_{\ell})$ contains all indices $m\leq k-1$.    Hence it is enough to consider $\ell = k$, that is, $j = p_k$.  We will therefore show that $ v_{p_k}(B_{p_k}) \geq v_{p_k}(B)$ for all $B \in \texttt{Storage}(k-1, a_{k})$. Note that, in Step~5, $a_k$ picks $|D_k|$ many bundles from equally-valued  (according to $a_{p_k}$) bundles in $\texttt{Storage}(k-1,a_{p_k})$.
  The trimming phase of the algorithm terminates when the domination condition in Step $9$ is satisfied for agent $a_{p_k}$.  That is, for the last round $T$ of the while loop we have $$  v_{p_k } (A_{p_k}^T) \geq v_{p_k} (A_{m}^T) + v_{p_k} (R^T)$$ for all $m \in D_k$, where $R^T$ is the residue at the beginning of round $T$.
  
  In the \texttt{Equaling} phase, even if the whole of $R^T$ is assigned to any bundle in $D_k$, the inequality $v_{p_k } (A_{p_k}) \geq v_{p_k} (A_{m})$ is ensured in the output allocation for all $m \in D_k$.
   \end{itemize}

\end{itemize}

\noindent
{\em Runtime:} The runtime analysis follows on similar lines as of  Lemma~\ref{lem:klogk complexity} by replacing $k$ with $|D_k|=d_k$, and therefore, we skip the analysis here.
\end{proof}

\section{Polynomial Query Complexity for Special Graph Structures } 
\label{sec:special graphs}
In this section, we take a close look at our techniques developed in the previous sections, and aim to optimize them with an overarching goal of identifying a non-trivial social graph structure over agents that admits a polynomial-query complexity for local envy-freeness. Developing efficient locally envy-free  protocols for interesting graph structures is listed as an open problem in \cite{bei2020cake} and \cite{AbebeKP16}. We partially address this open problem by developing  a novel protocol (\textsc{Alg2}) that finds a locally envy-free allocation among $n$ agents 
(a) on a \dTwoTree \ using only $O(n^3\log(n))$ \cut \ and $O(n^4\log(n))$ \eval \ queries (Section~\ref{sec:depth2}), and (b) on a $2$-\textsc{Star} using $O(n^2)$ \cut \ and $O(n^3)$ \eval \ queries (Section~\ref{sec:star}).

Before we begin, we will present a simple protocol \textsc{Alg1} (as a warm-up towards developing \textsc{Alg2}) that computes a locally envy-free allocation among four agents on a \textsc{Line}. We generalize and optimize the techniques developed for $\textsc{Alg1}$ with goal of attaining polynomial query protocol \textsc{Alg2} for \dTwoTree. Note that there are notable differences between the techniques of  $\textsc{Alg2}$ and our recursive protocol $\textsc{Domination}([0,1],1)$ developed in Section~\ref{sec:line}. We will discuss them later when we formally detail the $\textsc{Alg2}$ protocol.

\subsection{Four agents on a \textsc{Line} graph }\label{subsec:line}
 Consider a cake-division instance with four agents having $a_1-a_2-a_3-a_4$ as the underlying graph structure. We develop a simple protocol \textsc{Alg1} that returns a locally envy-free allocation among four agents on a \textsc{Line} by making $8$ \cut\  and $16$ \eval\ queries.
 
 \paragraph{Overview of \textsc{Alg1}:} In this protocol, we designate $a_3$ as the \emph{cutter} and $a_2$  as the \emph{trimmer} agent. Our protocol, \textsc{Alg1} runs in two phases: \texttt{Trimming} and \texttt{Equaling}.
  The first phase (\texttt{Trimming}) begins with  the cutter agent $a_3$ dividing the cake into four equally valued pieces according to her and $a_4$ selecting her favourite piece from them. Then,  $a_2$ makes her favorite piece (among remaining three pieces)  equal to the second largest using the \textsc{Trim}(.) procedure. The trimming obtained from this procedure is called as \emph{residue} and is allocated in the next phase. Similar to the first phase, in the second phase (\texttt{Equaling}), the cutter agent $a_3$ first divides the residue into four equal pieces according to her and  $a_4$ selects her favourite piece. However, unlike the first phase, instead $a_2$ makes her two most favorite pieces of equal value using  \textsc{Equal}(.) procedure. Finally, the pieces from these two phases are grouped together appropriately to form a complete partition of the cake $[0,1]$, and then are assigned to agents (see Step \ref{alg1:group} of \textsc{Alg1}).
 
The idea is to create a partition such that the cutter agent $a_3$ is indifferent between her and $a_4$'s bundle and values $a_2$'s bundle at most as much as her own bundle. The trimmer agent $a_2$ similarly is indifferent between her and $a_1$'s bundle whereas she values $a_3$'s bundle at most as valuable as her own bundle. \AlgOne \ ensures the above properties and hence finds a locally envy-free allocation (see Theorem \ref{thm:4LINE}).

\begin{theorem}
For cake-division instances with four agents on  a \textsc{Line} graph, there exists a discrete cake-cutting protocol that finds a locally envy-free allocation using   $8$ \cut \  queries and $16$ \eval \ queries. 
\label{thm:4LINE}
\end{theorem}

\begin{algorithm}[ht!]
    \renewcommand{\thealgocf}{}
    \DontPrintSemicolon
    \SetAlgorithmName{$\textsc{Algorithm}$}{ }{ }
 \nonl {\bf Input:} A cake-division instance $\mathcal{I}$ with four agents on a $\textsc{Line}$ \;
  \nonl {\bf Output:} A locally envy-free allocation \; 
  \nonl ------------\texttt{ Trimming Phase} ------------ \;
   Let $\{P_1, P_2, P_3,P_4\} $ be the pieces obtained when $a_3$ divides the entire cake ([0,1]) into 4 equal pieces \;
$a_4$ chooses her favorite piece, say $P_4$ \;
   Let $\{P_1,P_2\} \leftarrow \textsc{Select}(a_2, \{P_1, P_2, P_3\},2)$ \;
   Let $(\{P'_1,P_2\},T) \leftarrow \textsc{Trim}(a_2, \{P_1,P_2\})$  \texttt{/* $P'_1:= P_1 \setminus T, v_2(P'_1) = v_2(P_2)$ */} \;
\nonl ------------\texttt{ Equaling Phase} ------------ \;
  Let $\{T_1, T_2, T_3,T_4\} $ be the pieces obtained when $a_3$ divides the trimming $T$ into 4 equal pieces \;
   $a_4$ chooses her favorite piece, say $T_4$ \;
     Let $\{T_1,T_2\} \leftarrow \textsc{Select}(a_2, \{T_1, T_2, T_3\},2)$ \; 
     $\{T'_1, T'_2 \}$ = \textsc{Equal} $(a_2, 
     \{ T_1, T_2\})$ 
      \texttt{/* Let  $T'_1 = T_1 \setminus T'$ and $T'_2 = T_2 \cup T'$ such that $v_2(T'_1) = v_2(T'_2)$ */} \;
\nonl ---------------\texttt{ Create partition} -------------- \; 
  Set $A_1 := P'_1 \cup T'_2  , A_2 := P_2 \cup T'_1, A_3 := P_3 \cup T_3  $ and $A_4 := P_4 \cup T_4$\; \label{alg1:group}
\nonl  ---------------\texttt{  Final Allocation} -------------- \;
     $A_4$  is allocated to $a_4$ and $A_3$ is allocated to $a_3$ \;
   $a_1$ picks her favourite piece between  $A_1$ and $A_2$. The remaining piece is allocated to  $a_2$\;
       	\caption{Local Envy-freeness for four agents on a $\textsc{Line} \ (\AlgOne) $}
 	\label{alg:FourAgentsOnPath}
\end{algorithm} 

\begin{proof}
To begin with, note that \textsc{Alg1} returns a complete allocation $\{A_1,A_2, A_3,A_4\}$ of the cake, since \textsc{Equal} procedure allocates the entire residue from the first phase (see Step 8). 


Next, we will establish local envy-freeness. Recall that agents have additive valuations. Therefore, since $a_3$ is the cutter, she values her piece, $A_3 = P_3 \cup T_3$ and $a_4$'s piece, $A_4 = P_4 \cup T_4$ equally i.e., $v_3(A_3) = v_3(A_4)$. Furthermore,
Steps $2$ and $6$ ensures that $a_4$ values her piece at least as much as that of $a_3$, i.e., $v_4(A_4) \geq  v_4(A_3)$.  That is, there is no envy between agents $a_4$ and $a_3$.
 Next, to see that there is no envy between $a_2$ and $a_3$, observe 
 \begin{flalign*}
    v_3(A_3) = v_{3}(P_3) + v_3(T_3) &=  v_{3}(P_2) + v_{3}(T_1) \tag{by Steps $1$ and $5$ of \textsc{Alg1}} \\
    &\geq  v_{3}(P_2) + v_3(T'_1) = v_{3}(A_2), \tag{since $T'_1 \subseteq T_1$}  
     \end{flalign*}
     and 
     \begin{flalign*} 
     v_3(A_3) = v_{3}(P_3) + v_3(T_3) &\geq   v_{3}(P_1) + v_3(T_2)  \tag{by Steps $1$ and $5$ of \textsc{Alg1}}\\
    &\geq  v_{3}(P'_1) + v_3(T'_2)= v_3(A_1). \tag{since $P_1 \cup T'_2 \subseteq P_1$}
      \end{flalign*}
That is, irrespective of which bundle $a_2$ gets  between $A_1$ and  $A_2$, agent $a_3$ does not envy $a_2$. Also, $a_2$ does not envy $a_3$ since she selects her top two pieces in Steps 3 and 7, that is 
 \begin{equation}
    v_2(A_1) = v_{2}(P'_1) + v_2(T'_2) =  v_{2}(P_2) + v_{2}(T'_1) = v_2(A_2) \geq v_{2}(P_3) + v_2(T_3) =v_{2}(A_3). 
    \label{eq: inequalities for a2}
     \end{equation}

Equation (\ref{eq: inequalities for a2}) implies that $a_2$ values $A_1$ and $A_2$ equally, i.e.,   $v_2(A_1)=v_2(A_2)$. Hence, $a_2$ does not envy $a_1$ irrespective of which bundle she gets between $A_1$ and $A_2$.  
 Finally,  $a_1$ does not envy $a_2$ since she picks her favorite bundle between $A_1$ and $A_2$ in Step 11, thereby establishing local envy-freeness.\\

\noindent
\textit{Counting queries:} Steps $1$ and $5$ each require $3$ \cut \ queries (for the cutter agent $a_3$) to divide the cake $[0,1]$ and the residue $T$, respectively, into four equi-valued pieces. For the trimmer agent $a_3$, we require $1$ \cut\ query to execute Step $4$ and $1$ more \cut\ query to execute Step $8$, making a total of $8$ \cut\ queries. Next, observe that $\AlgOne$ requires a total of $16$ \eval \ queries: a sum of $7$ queries ($3$ and $4$ queries in Steps $2$ and $6$ respectively) for agent $a_4$ to select her favorite piece, $6$ queries (3 queries each in Steps $3$ and $7$)  for agent $a_2$  to select her two highest-valued pieces, $1$ query for agent $a_3$ to evaluate her value for $T$ in Step $5$, and $2$ queries for agent $a_1$ in Step $11$ to select her favorite piece. 
\end{proof}

\subsection{Local Envy-freeness on \textsc{Tree} graphs with depth at-most 2}
\label{sec:depth2}
In this section, we consider cake-division instances where agents lie on a social graph that is a tree with depth at-most $2$ (\dTwoTree), and show that it admits a simple and query-efficient  protocol (\textsc{Alg2}) for finding a local envy-freeness (see Theorem \ref{thm:depth2tree}). The idea is to efficiently generalize the key elements of \textsc{Alg1}. Therefore, similar to  \textsc{Alg1}, it consists of \texttt{Trimming} and \texttt{Equaling} phases, but it will be much more involved since the root agent has multiple children who are connected to multiple leaf agents, as opposed to the case of four agents on a \textsc{Line}. It is relevant to note that while there are certain similarities with our previous protocol $\textsc{Domination}([0,1],1)$, \textsc{Alg2} is \emph{not} a recursive algorithm.  

We write $a_r$ to denote the root agent and $D$ to denote the set of her neighbours/children (see Figure~\ref{one}). Each agent $a_i \in D$ has $\ell_i +1$ neighbours, i.e., she is connected to $\ell_i \geq 0 $ leaf agents. In addition, we write $L(i)$ to denote the set of children of agent $a_i \in D$. For cake-division instances, we will specify a \dTwoTree \ by $\langle n, a_r, D, \{\ell_i\}_{a_i \in D} \rangle$.

Our protocol \textsc{Alg2} designates $a_r$ as the cutter and each  $a_i \in D$ as the trimmer agents (recall the definition from Section \ref{sec:setting}). For clarity, we will call the subset of  trimmer agents who perform \textsc{Trim}(.) procedure as \emph{active}  trimmer agents, denoted by the set $Tr$ and the agents who perform \textsc{Equal}(.) procedure as \emph{equalizer} agents.  
 Before giving an overview of \textsc{Alg2}, we first state the \emph{domination} condition used in \textsc{Alg2}.
 \begin{figure}[ht!]
 \centering 
 \includegraphics[scale =0.45]{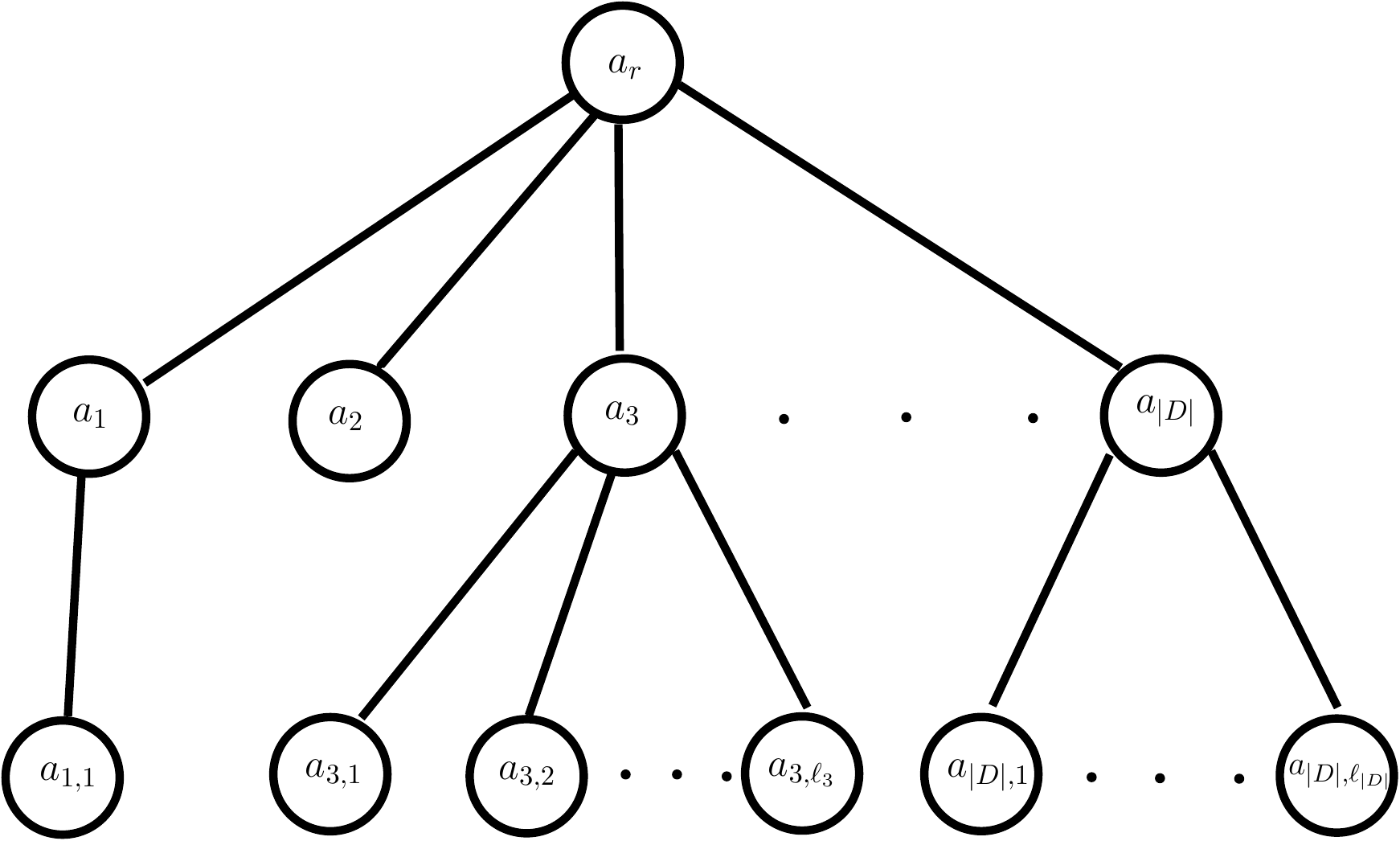}
 \caption{\textsc{Tree} with depth at-most two (\textsc{Depth2Tree}).  }
  \label{one}
 \end{figure}
 
\begin{definition}[Domination Condition] \label{domcond}
Given a cake-division instance on a  \textsc{Depth2Tree}, let us denote root agent's bundle by $A_r$ and the residue by $R$. We say that the root agent \emph{dominates} a bundle $A_i$ of agent $a_i$ if
\begin{equation}
\label{eq:domination}
v_r(A_r) - v_r(A_i) \geq \min \{ \frac{\ell_i +1}{|D| +1},1\}  \cdot v_r(R). \end{equation}  
\end{definition}

 The protocol \textsc{Alg2} seeks to achieve \emph{dominance} for the root agent on all (the bundles maintained by) her children by iteratively dividing the currently unallocated cake (to be referred as \emph{residue}) in multiple rounds of a while-loop. In the beginning of the protocol, the set of active trimmer agents $Tr=D$ and the algorithm progressively removes agents from $Tr$ by checking the dominance condition.\footnote{An agent $a_i \in D$ is removed from the set $Tr$ of active trimmer agents as soon as the root agent starts dominating her.} And, the while-loop gets terminated when the set $Tr$ becomes empty.
 
 \paragraph{An Overview of  \textsc{Alg2}:} The protocol iteratively constructs an allocation $\mathcal{A}=\{A_1, \dots,A_n\}$ of the cake among $n$ agents. It primarily consists of a while-loop (Steps 3-22) whose goal is to establish dominance of agent $a_r$ over her children in $D$ in multiple rounds. We call $R$ as \emph{residue} of the cake (that is yet unallocated) obtained at the end of each round of the while-loop. In the beginning, $R=[0,1]$, and its value according to $a_r$ keeps decreasing with subsequent rounds of the while-loop (see Lemma~\ref{lemma:decreaseinround}).
 
 There are three phases in each round, namely \texttt{Selection}, \texttt{Trimming}, and \texttt{Equaling}. Consider any round $t$ of the while-loop and denote the residue at the beginning of this round as $R^t$. The round starts with agent $a_r$ dividing the current residue $R^t$ into $n$ equal pieces, followed by \texttt{selection} of $\ell_i+1$ most favorite available pieces by every agent $a_i \in D$ (Step~5). The last remaining piece in added to the root agent's bundle $A_r$. For each agent $a_i \in D$, throughout the execution, the algorithm maintains a collection $\mathcal{A}^{(i)} = \{A^{(i)}_0, \dots, A^{(i)}_{\ell_i}\}$ of $\ell_i+1$ bundles so that each of them are of equal value according to $a_i$. 

In every round of the while-loop, post the \texttt{selection} phase, there are possibly two phases for an agent $a_i \in D$: \texttt{trimming} if agent $a_r$ does not \emph{dominate} her yet and \texttt{equaling} if $a_r$ has started to \emph{dominate} her. If an agent $a_i \in D$ is in \texttt{trimming} (\texttt{equaling}) or phase, she performs a \textsc{Trim} procedure (\textsc{Equal} procedure, respectively) on the set of bundles she selected in Step 5.
Moreover, for each agent $a_i \in D$, the protocol maintains a special bundle  $A^{(i)}_0 \in \mathcal{A}^{(i)}$ which receives a whole (untrimmed) piece in each step of the \texttt{trimming} phase (see Step 11). Towards the end of the \texttt{trimming} phase, the algorithm assigns each of the $\ell_i$ trimmed pieces to $\ell_i$ bundles in $\mathcal{A}^{(i)} \setminus A^{(i)}_0$. 
We remark here that, with each round, this operation (see Steps 12-14) increases the difference between root agent's valuation for her own bundle and any of these bundles containing the trimmed pieces, see Lemma \ref{lemma:decreaseinround}. This proves to be crucial in eventually establishing the required dominance for agent $a_r$.

In Lemma~\ref{lemma:domination}, we prove that after every $O(n \log n)$ rounds of the while-loop, the root agent $a_r$ starts dominating a new trimmer agent from the set $Tr \subseteq D$. This is when the protocol remove this trimmer agent from the set $Tr$ of active trimmer agents and she performs \textsc{Equal}(.) procedure in the next round (see Step 21). Therefore, we obtain that the set of active trimmer agents become empty (i.e., every agent $a_i \in D$ becomes an equalizer) in polynomially-many rounds.

In the \texttt{equaling} phase, agent $a_i$ makes all the pieces (picked in the \texttt{selection} phase) of equal value to her and then allocates them to the bundles in $\mathcal{A}^{(i)}$ such that each bundle gets a single piece. Bundle $A^{(i)}_0$ now receives a trimmed piece so that, due to dominance, the agent $a_r$ will not envy its neighbour even if she is allocated this bundle (see Steps 16-19). At the end, each leaf agent $a \in L_i$ chooses her favorite bundle from $\mathcal{A}^{(i)}$ and agent $a_i$ receives the last remaining bundle. This creates a complete allocation $\mathcal{A}$ of the cake.

\begin{algorithm}[ht!]
 \renewcommand{\thealgocf}{}
    \DontPrintSemicolon
     \SetAlgorithmName{$\textsc{Alg2}$}{ }{ }
    \caption{Local Envy-freeness  for $n$ agents on  \textsc{Depth2Tree}}
        \label{alg:2-tree}
       \nonl \textbf{Input:} A cake-division instance $\mathcal{I}$ on  \textsc{Depth2Tree} $(n, a_r, D, \{\ell_i\}_{a_i \in D})$.\;
       \nonl \textbf{Output:} A locally envy-free allocation. \;
        Initialize  $R \leftarrow [0,1]$, set of trimmer agents $Tr \leftarrow D$, bundles $A_0^{(i)},\dots,A_{\ell_i}^{(i)} \leftarrow \emptyset$ for each $a_i \in D$ and a bundle $A_r \leftarrow \emptyset$ for the root agent.\;
        \While{$Tr \neq \emptyset$}
        {
          $\mathcal{X} \gets \textsc{Eq-Div}(a_r, R,n)$\;
         \nonl  --------$\texttt{ Selection}$--------\\
         
         \For{$a_i \in D$}
         {
            Set $(\mathcal{X}^{(i)},\mathcal{X}) \leftarrow$  \textsc{Select}($a_i,\mathcal{X},\ell_{i}+1$)
         }
          Set $A_r \gets A_r \cup \mathcal{X}$\;
          \nonl
 --------$\texttt{Trimming}$--------- \\
         Set $R \gets \emptyset$\;
         \For{ $a_i\in Tr$}
         {
         Let $X^{(i)}_0 = \argmin_{X \in \mathcal{X}^{(i)}}v_{i}(X)$\tcc{This piece won't be trimmed}
            $(\mathcal{X}^{(i)},R) \gets \textsc{Trim}(a_i,\mathcal{X}^{(i)})$\;
            $A_0^{(i)}\leftarrow A_0^{(i)} \cup X_{0}^{(i)}$ and $\mathcal{X}^{(i)}\leftarrow \mathcal{X}^{(i)}\setminus X_{0}^{(i)}$\; \label{line:ai0}
            Let $A_w^{(i)} = \argmax_{1\le k\le \ell_i } v_{r}(A^{(i)}_k)$ and $X_t^{(i)} = \argmin_{1\le k\le \ell_i}v_{r}(X_k^{(i)})$\; \label{inalg-argmin}
             $A_w^{(i)}\leftarrow A_w^{(i)}\cup X_t^{(i)}$ and $\mathcal{X}^{(i)}\leftarrow \mathcal{X}^{(i)}\setminus X_t^{(i)}$
             \tcc{Trying to achieve domination on $A_w^{(i)}$ for the root agent}\label{inalg-cupmax}
            For each $k\neq 0, w$, add one arbitrary piece from $\mathcal{X}^{(i)}$  to $A_k^{(i)}$ \;
         }

 \nonl 
 --------$\texttt{Equaling}$--------- \\
         \For{$a_i\in D\setminus Tr$}
         {
            $(\mathcal{X}^{(i)}, X^*) \gets \textsc{Equal}(a_i,\mathcal{X}^{(i)})$\;
            Let $X_{k}^{(i)} \in \mathcal{X}^{(i)}$ be the piece such that $X_{k}^{(i)} \subseteq X^*$
            \tcc{There is only one piece satisfying this condition.}\label{inalg-asubpiece}
            $A_{0}^{i} \gets A_{0}^{i}\cup X_{k}^{(i)}$\tcc{This ensures that the root agent will not envy the bundle $A_{0}^{i}$}\label{inalg-unionsub}
           For each $k \neq 0$, add one arbitrary piece from $\mathcal{X}^{(i)}$  to the bundle $A_k^{(i)}$ \;
         }
         \nonl
         --------$\texttt{Judging Domination}$-------- \\
         \For{$a_i \in Tr$}
         {
         \uIf{$\forall k>0, v_{r}(A_r)-v_{r}(A_k^{(i)})\ge \min\{\frac{\ell_{i}+1}{|D|+1},1\}\cdot v_{r}(R)$\label{inalg-domination}}
         {$Tr\gets Tr\setminus \{i\}$}
         }
         }
          \nonl --------$\texttt{Choose after while loop}$-------- \\
        \For{$a_i\in D$}
        {
            \For{$a_j\in L(i)$} 
            {
            $a_j$ is allocated her favorite (available) bundle from $\{A^{(i)}_{k}\}_{0 \leq k \leq \ell_{i}}$ \;
            }
            $a_i$ is allocated the remaining bundle\;
        }
        \Return  The allocation $A_r \cup \{A_j^{(i)}\}_{0\leq j\leq \ell_i, i \in D}$\;
\end{algorithm}

We begin by establishing three crucial properties of \textsc{Alg2} in Lemmas \ref{lemma:O(n)cuts}, \ref{lemma:decreaseinround} and \ref{lemma:domination}. These lemmas would be useful in proving our main results in  Theorems \ref{thm:depth2tree} and \ref{thm:2star}.

\begin{lemma}\label{lemma:O(n)cuts}
In every round of the while-loop in $\textsc{Alg2}$, we make $O(n)$ many cuts on the cake. 
\end{lemma}
\begin{proof}
In any round of while-loop in the execution of \textsc{Alg2}, there are three steps where cuts are made on the cake. First, the root agent makes $n-1$ cuts in Step $4$ to equally divide $R$ into $n$ pieces.  
In the \texttt{Trimming} phase, each agent $a_i$ makes $\ell_i$ many cuts in Step $11$. In total, it will be no more than $n$ cuts. 

In the \texttt{Equaling} phase, each agent $a_i$ requires at most $\ell_i$ cuts to execute Step $17$. For each piece larger than the average, it requires one cut to make it equal to the average.  For each piece less than the average, we add some pieces and make at most one cut. In total, there are no more than $n$ cuts. Therefore, in total there are $O(n)$ cuts in each while loop of \textsc{Alg2}. 
\end{proof}

We write $R^t$ to denote the residue at the beginning of round $t$ of the while-loop in \textsc{Alg2}.

\begin{lemma}\label{lemma:decreaseinround}
Consider any round $t$ of the while-loop in \textsc{Alg2}. We have the following bound on the valuation of the root agent for the residues from two consecutive iterations of the while-loop, $$v_r(R^{t+1})\le \left(1-\frac{|D|+1}{n}\right)v_r(R^t)$$ In other words, the residue is decreasing exponentially with respect to the root agent with each iteration of the while-loop.
\end{lemma}
\begin{proof}
Any round $t$ of the while-loop begins with the root agent dividing the residue $R^t$ into $n$ equal pieces, each of value $v_r(R^t)/n$. 

Recall that $D$ is the set of the neighbours of $a_r$. We will prove that there are at least $|D|+1$ whole pieces that do not generate any residue. The root agent chooses one whole piece. For each agent $a_i\in D$, that is still a trimmer agent, she would reserve a whole piece for $A_0^{(i)}$. Otherwise, for the remaining agents in the set $D$, there is no residue in the \texttt{Equaling} phase. Hence, each agent $i\in D$ and the root agent keep at least one whole piece without generating any residue. Thus, there are at least $\frac{|D|+1}{n}$ proportion less residue towards the end of round $t$ of the while-loop. Therefore, it follows that  $v_r(R^{t+1})\le(1-\frac{|D|+1}{n})v_r(R^t)$.
\end{proof}

\begin{lemma}\label{lemma:domination}
Consider any round $t$ of the while-loop in \textsc{Alg2}. According to the root agent $a_r$, the total value of all the pieces selected by any agent $a_i \in D\setminus Tr$ in rounds $\ell \geq t$ is at most $\min\{\frac{\ell_i+1}{|D|+1},1\}\cdot v_r(R^t)$. 
\end{lemma}
\begin{proof}
Consider some round $t$ of the while-loop in $\textsc{Alg2}$. Fix any agent $a_i \in D \setminus Tr$. She executes the \texttt{Equaling} procedure and does not generate any residue. Therefore, it follows trivially that according to $a_r$, the total value of all the pieces selected by $a_i$ in round $\ell \geq t$ is at most $R^t$.

Next, at the beginning of round $t$, $a_r$ divides the residue $R^t$ into $n$ equal pieces according to her and
$a_i$ selects $\ell_{i}+1$ out of them. By Lemma \ref{lemma:decreaseinround}, we know that $v_r(R^{t+1})\le(1-\frac{|D|+1}{n}) \cdot v_r(R^t)$. Therefore, the total value of all the pieces selected by $a_i$ in the subsequent rounds (including round $t$) is at-most
$$\frac{\ell_{i}+1}{n}\sum_{m=t}^{\infty}\left(\left(1-\frac{|D|+1}{n}\right)^{m-t} \cdot v_r(R^t)\right)\leq\frac{\ell_i+1}{|D|+1}\cdot v_r(R^t).$$
This proves the stated claim.
\end{proof}
We now establish that {\sc Alg2} finds locally envy-free allocation among $n$ agents on   \dTwoTree.
\Trees*

\begin{proof} We begin by proving the correctness of \textsc{Alg2}, by arguing that each type of agent is locally envy-free.\\
 
 \noindent
 (a) \textbf{Root agent:} In each round of the while-loop, root agent $a_r$ divides the current residue into $n$ equal pieces, denoted by the set $\mathcal{X}$. For any agent $a_i\in Tr$, the piece added to any bundle $A^{(i)}_k \in \mathcal{A}^{(i)}$ is a subset of a piece from $\mathcal{X}$. Therefore, $a_r$ will not value any bundle $A^{(i)}_k$ (for $k \geq 0)$ larger than her own bundle $A_r$.
        When $a_r$ starts dominating an agent $a_i \in Tr$, then it is removed from $Tr$ and is therefore in the set $D\setminus Tr$. For $k\ge 1$, the bundles $A^{(i)}_k$ (for $k \geq 1$) must therefore satisfy the domination condition (line \ref{inalg-domination}) in the round in which agent $a_i$ is removed from the set $Tr$, i.e., $v_r(A_r)-v_r(A^{(i)}_k)\ge\min\{\frac{\ell_i+1}{|D|+1},1\} \cdot v_r(R)$.
    Note that if $v_r(A_r)-v_r(A^{(i)}_k)\ge v_r(R)$, no matter how the residue is allocated, the root agent will not value $A^{(i)}_k$ higher than $A_r$. On the other hand, if $v_r(A_r)-v_r(A^{(i)}_k)\ge \frac{\ell_i+1}{|D|+1}\cdot v_r(R)$, Lemma \ref{lemma:domination} ensures that the root agent again prefers $A_r$ to $A^{(i)}_k$.    
    
    Finally, observe that the algorithm ensures that the bundle $A^{(i)}_0$ is a subset of a piece from $\mathcal{X}$ (see Steps~\ref{line:ai0}, \ref{inalg-asubpiece} and \ref{inalg-unionsub}). Therefore, the root agent will not envy any of her neighbours in the output allocation.

    \noindent (b) \textbf{Neighbour agents:} For any agent $a_i\in D$, every bundle $A^{(i)}_k \in \mathcal{A}^{(i)}$ is of equal value in the view of agent $a_i$. So no matter how the leaf agents (that are her neighbours) choose, agent $a_i$ will have no envy towards them. In the \texttt{selection} phase, agent $a_i$ chooses $\ell_i+1$ of her favorite pieces from $\mathcal{X}$, before the root agent. In each round of the while-loop, the increment for each bundle $A^{(i)}_k$ for $k \geq 0$ is as large as the increment in bundle $A_r$ in the view of agent $a_i$. So agent $a_i$ will not envy the root agent. Therefore, there is no envy for agent $a_i$ in the final allocation.
    
    \noindent (c) \textbf{Leaf agents:} Every leaf agent chooses her favorite bundle before the neighbour agent she is connected to.
    
Therefore, \textsc{Alg2} outputs a locally envy-free allocation among $n$ agents on a \textsc{Depth2Tree}. 


 \noindent
Next, let us analyse the query complexity of \textsc{Alg2}. We begin by proving that in every $O(n\log n)$ rounds, there must be an agent in the active trimmer set $Tr$ that the root agent starts dominating. The following claim serves as a prerequisite for proving the above statement. 
\begin{claim}\label{lemma:rounds}
Let us denote $c_t:=\max_{a_i\in Tr, k\le \ell_i}\{\frac{v_r(R^t)}{n}-v_r(X^{(i)}_k)\}$ to be the maximum trimmed value in round  $t$ of the while-loop in \textsc{Alg2}. After $O(n\log n)$ rounds of the while-loop, we have $v_r(R^{t+O(n\log n)})\le c_t$.
\end{claim}
\begin{proof}
As there are at most $n$ trimmed pieces and $c_t$ is the maximum trimmed value in round $t$ of the while-loop, we have $v_r(R^{t+1})\le n \cdot c_t$.

As $|D|+1\ge 2$, by Lemma \ref{lemma:decreaseinround}, the residue satisfies the inequality $v_r(R^{d+1})\le (1-2/n)v_r(R^d)$ for any two consecutive rounds $d$ and $d+1$. Therefore, after $O(n\log n)$ rounds, we obtain the desired bound on $a_r$'s value for the residue $v_r(R^{t+O(n\log n)})\le n\cdot c_t \cdot (1-2/n)^{O(n\log n)}\le c_t$.
\end{proof}


\begin{lemma} \label{lemma:domination}
If $Tr\neq\emptyset$, then after $O(n\log n)$ rounds, there is an agent $a_i\in Tr$ that gets removed from the set, i.e., the root agents starts dominating $a_i$.
\end{lemma}
\begin{proof}
Recall that, \textsc{Alg2} removes an agent $a_i$ from the set $Tr$ when the root agent starts dominating her; see the domination condition (\ref{domcond}). Consider round $t$ of the while-loop in \textsc{Alg2} with residue $R^t$ at its beginning.
Let us write $c_t:=\max_{i\in Tr, k\le \ell_i}\{\frac{1}{n}v_r(R^t)-v_r(X^{(i)}_k)\}$ to denote the maximum trimmed value in this round. Let $(i^j,k^j)$ be the pair such that $\frac{1}{n}v_r(R^j)-v_r(X^{(i^j)}_{k^j})=c_t$.
 Since $\sum_{a_i \in D} \ell_i \leq n$, therefore, by the pigeon hole principle, between the rounds $j$ to $j+n$, there must be an agent $a_{i'} \in Tr$ that appears at least $\ell_{i'}$ many times in the pairs $\{(i^h,k^h)\}_{j\le h\le j+n}$.

We will prove that agent $a_{i'}$ will be removed in round $t+O(n\log n)$. Let us define $K=\{k\mid (i',k)=(i^h,k^h)\text{ for some round } t\le h\le t+n\}$ to be the set of indices of the bundles of agent $a_{i'}$ that corresponds to the maximum trimmed value in rounds $t$ to $t+n$. By Lemma \ref{lemma:rounds}, it follows that for any $k\in K$, the root agent will start dominating the piece $A^{(i')}_k$ after $O(n\log n)$ rounds (from round $t$). If $\{1, \dots, \ell_{i'}\} \subseteq K$, then all the bundles $\{A^{(i')}_k\}_{1\le k\le \ell_{i'}}$ are dominated by the root agent, and therefore we obtain that the root agent dominates agent $a_{i'}$ after $O(n\log n)$ rounds.

Otherwise, there exists some index between $1$ to $\ell_{i'}$ that is missing in $K$. This implies that there exists some $k'\in K$ that appears at least twice in pairs $\{(i^h,k^h)\}_{t\le h\le t+n}$. Let $h_1<h_2$ be two indices such that $(i^{h_1},k^{h_1})=(i^{h_2},k^{h_2})=(i',k')$. After round $h_1$, we have $v_r(A_r)-v_r(A^{(i')}_{k'})\ge c_{h_1}$. In  round $h_2$, since the index $k'$ gets picked again, we have $v_r(A^{(i')}_{k})\le v_r(A^{(i')}_{k'})$ for $1\le k\le \ell_{i'}$ (by Steps \ref{inalg-argmin} and \ref{inalg-cupmax}). Therefore, we obtain that the difference $v_r(A_r)-v_r(A^{(i')}_{k})\ge c_{h_1}$ for all $k$ in round $h_2$. Hence, by Lemma \ref{lemma:rounds}, all of the pieces will be dominated after $O(n\log n)$ rounds, proving the stated claim.
\end{proof}

The above lemma therefore proves that, in every  $O(n\log n)$ rounds, the size $|Tr|$ of the active trimmer agents will decrease at least by one. Hence, the while-loop will end after $O(n^2\log n)$ many rounds. By Lemma \ref{lemma:O(n)cuts}, we know that each round of the while-loop requires $O(n)$ cuts. Hence, we conclude that \textsc{Alg2} requires we have a total of $O(n^3\log n)$ \cut \ queries. For the \eval \ queries, the worst case is that each agent evaluates all the contiguous pieces. Because there are $O(n^3\log(n))$ cuts, there are $O(n^3\log(n))$ contiguous pieces. Hence, the algorithm makes $O(n^4\log(n))$ \eval~queries. Therefore, \textsc{Alg2} finds a locally envy-free allocation among $n$ agents on a \dTwoTree  \ using the stated number of queries, therefore completing the proof.
\end{proof}

\subsection{Local Envy-freeness on   \textsc{Star} Graphs} \label{sec:star}

In this section, we analyse our protocol (\textsc{Alg2}) developed in the previous section for a special case of graph structure 2-$\textsc{Star}$ graphs (\dTwoTree\ where every non-root agent is connected to at most two agents). We will show that it only requires $O(n^2)$ \cut \ queries and $O(n^3)$ \eval \ queries to find a locally envy-free allocation among $n$ agents on a  2-\textsc{Star} graph (see Theorem \ref{thm:2star}). 
\begin{figure}[ht!]
  \centering
    \includegraphics[width=0.4\textwidth]{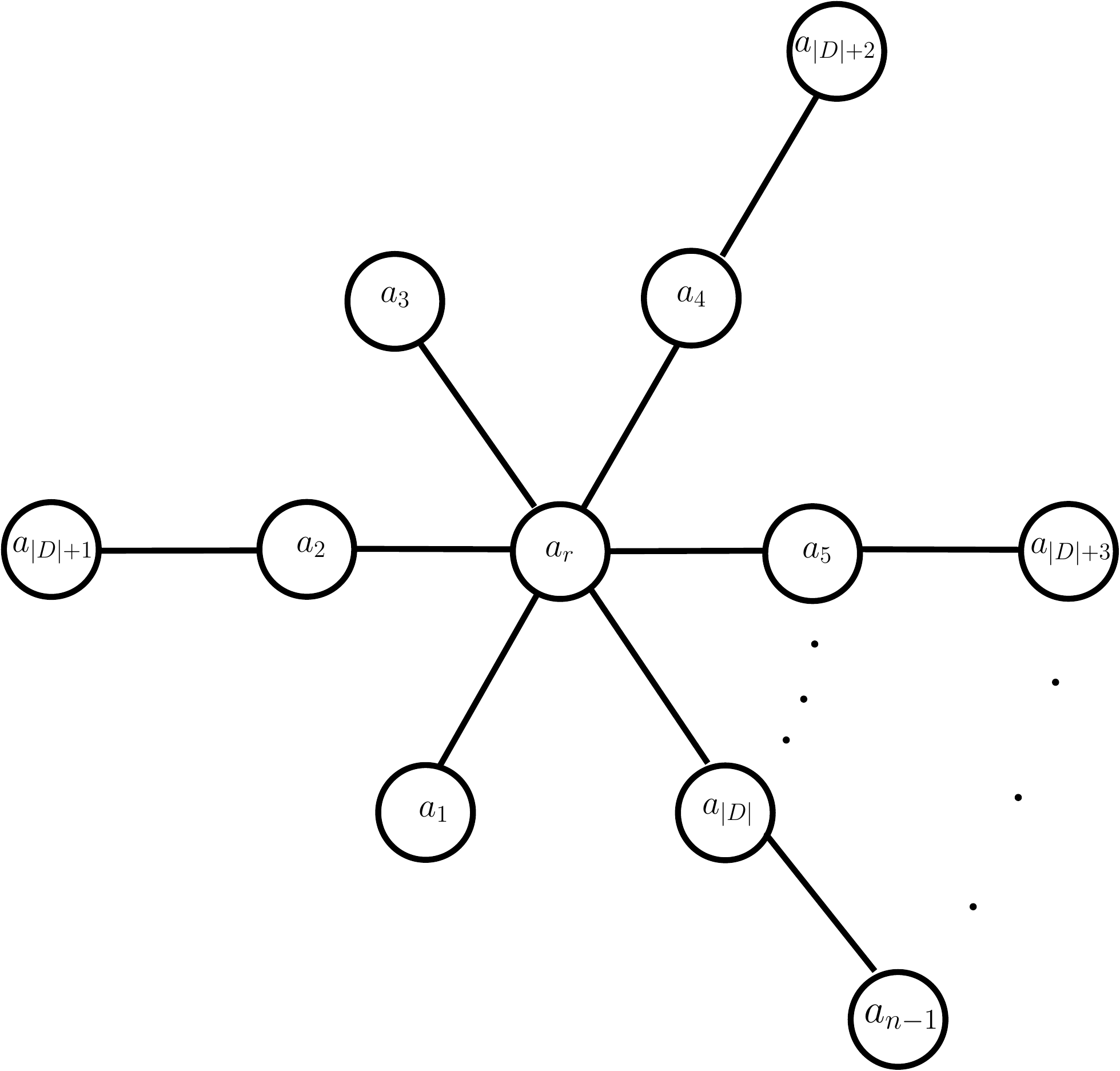}
  \caption{Illustration of a \textsc{2-Star} graph}
    \label{fig:twostar}
\end{figure}


 First, observe that a simple cut and choose protocol produces a locally envy-free allocation for  \textsc{Star} graphs. The root agent cuts the cake into $n$ equal pieces and the remaining agents pick their favourite (available) piece one after the other while the last remaining piece is allocated to the root agent. It is easy to see that this protocol  produces a locally envy-free allocation with $n-1$ \cut\  and $O(n^2)$ \eval\ queries. 

Now, let us focus on 2-\textsc{Star} graphs. Using the exact same protocol (\textsc{Alg2}) as used in Section~\ref{sec:depth2} for \dTwoTree, we will show a significantly lower query complexity for \twoStar\ graphs. This lower query complexity for  $2$-\textsc{Star} graphs can be attributed to the following two structural facts: (a)  a 2-\textsc{Star} graph has a large number of trimmer agents ($|D| \geq \frac{n-1}{2}$), which makes the residue decrease faster, and (b) the  trimmer agents  are connected to at-most one leaf agent ($\ell_i \leq 1$). We state and prove the main result of this section.
\TwoStar*

\begin{proof}
We direct the reader to refer to  the proof of   Theorem \ref{thm:depth2tree} for  local envy-freeness guarantee of   \textsc{Alg2} on  \textsc{2-Star} graph as it is a special case of \textsc{Depth2Tree}. We will analyze the number of \cut\ and \eval\ queries required by \textsc{Alg2} for finding a locally envy-free allocation in a 2-\textsc{Star} graph. We begin by proving that there are $O(n)$ many cuts made on the cake in any execution of the while-loop and  then, we will show that the while-loop ends in $n$ rounds.


\begin{lemma}\label{lemma:stopn}
When the set of trimmer agents is non-empty ($Tr\neq\emptyset$), at least one agent gets removed from the set $Tr$ in every two rounds of the while-loop.
\end{lemma}
\begin{proof}
Consider a 2-\textsc{Star} graph with the root agent having $|D|$ many neighbours and $\ell$ leaf agents. That is, there are a total of $n = |D|+\ell+1$ agents. Each neighbour agent is connected to at most one leaf agent. Therefore, the total number of leaf agents, $\ell \leq |D|$ and hence $|D|+1 \geq n/2$.
Then, by Lemma \ref{lemma:decreaseinround}, the residue will decrease at least by half in each round. 

Consider some round $t$ of the while-loop, with $R^t$ residue at its beginning. We will identify a special trimmer agent related to $R^{t+1}$ who can be removed from the set $Tr$ at the end of round $t+1$ of the while-loop. Recall that any agent $a_i \in Tr$, she selects two pieces $X^{(i)}_0$ and $X^{(i)}_1$ in the \textsc{Selection} phase.\footnote{In Step $10$, we rename these pieces such that $v_i(X^{(i)}_0) \leq v_i(X^{(i)}_1)$.}
Let us denote $c_t=\max_{a_i \in Tr}\{\frac{v_r(R^t)}{n}-v_r(X^{(i)}_1)\}$ be the maximum trimmed value in round $t$. Let $a_i^* \in Tr$ be the agent who did the trim for the maximum trimmed value, i.e., after the trim, we have $\frac{v_r(R^t)}{n}-v_r(X^{(i^*)}_1)=c_t$. Since a piece of value $\frac{v_r(R^t)}{n}$ is added to the bundle $A_1^{(i^*)}$ and $X^{(i^*)}_1$ is added to the bundle $A_1^{(i^*)}$, therefore we have that the difference $v_r(A_r)-v_r(A_1^{(i^*)})\ge c_t$.

Next, we will prove that the residue at the beginning of round $t+1$ is $v_r(R^{t+1}) \le\frac{n}{2}\cdot c_t$. Note that there are at most $n/2$ leaf agents, therefore at most $n/2$ pieces get trimmed and each trimmed part is of value at most $c_t$ (according to $a_r$). Therefore, we obtain $v_r(R^{t+1})\le\frac{n}{2}\cdot c_t$.

Finally, let us take a look at what happens in rounds $t$ and $t+1$. At the end of the round $t$, we have $v_r(A_r)-v_r(A_1^{(i^*)})\ge c_t$ and  $v_r(R^{t+1})\le\frac{n}{2}\cdot c_t$. We also know that (by Lemma \ref{lemma:decreaseinround}), at the end of round $t+1$, the residue $R^{t+2}$ has value $v_r(R^{t+2})\le \frac{1}{2} v_r(R^{t+1})\le\frac{n}{4} \cdot c_t$. The domination condition for a trimmer agent to get removed from the set $Tr$ in round $t+2$ is
$$v_r(A_r)-v_r(A_1^{(i^*)})\ge\frac{\ell_i}{|D|+1} \cdot v_r(R^{t+2})\ge\frac{4}{n} \cdot v_r(R^{t+2})$$ 
Note that, the difference on the left hand side is $v_r(A_r)-v_r(A_1^{(i^*)})\ge c_t$ and the value in the right hand side is $\frac{4}{n} \cdot v_r(R^{t+1})\le c_t$. Therefore, the domination condition for agent $a_{i^*}$  is satisfied at the end of round $t+1$ of while-loop, and hence she is removed from the set $Tr$ of trimmer agents.
\end{proof}

Therefore, Lemma \ref{lemma:stopn} proves that the while-loop will end in $2n$ rounds, and each round requires $O(n)$ cuts (by Lemma \ref{lemma:O(n)cuts}). Therefore, we have $O(n^2)$ cuts in total. 
For the evaluation queries, the worst case is that each agent evaluate all the contiguous pieces. Because there are $O(n^2)$ cuts, there are $O(n^2)$ contiguous pieces. We have at most $O(n^3)$ \eval~queries. 
\end{proof}

\section{Discussion and Future Directions}
This work explores the classic enny-free cake cutting  problem where agents lie on a social graph specifying the envy constraints between them. We develop a novel algorithm that computes a locally envy-free allocation among $n$ agents on any \textsc{Tree} with $n^{O(n)}$ queries under Robertson-Webb query model, thereby significantly improving the best known hyper-exponential upper bounds. We use our techniques developed above and optimize them to develop polynomial query complexity protocols for \dTwoTree\ and \twoStar\ graphs with  $O(n^4\log n)$ and $O(n^3)$ query complexity respectively.   To the best of our knowledge, our work is the first to identify non-trivial social graphs over agents that admits query-efficient discrete protocols for local envy-freeness. 

We believe that exploring the classic problem of envy-free cake division with a lens of social graphs among agents imparts novel insights and help understand the bottleneck in its computational complexity. There are many possible future research directions arising from our work. Here, we list some of them:
\begin{enumerate}
\item For trees with constant depth as an underlying social graph, can we develop query-efficient protocols by generalizing \textsc{Alg2} that works for trees with depth at most two?
\item Can we find a simple and efficient protocol for social graphs such as a cycle or a bipartite graph? 
\item Can we develop a parameterized protocol with respect to the tree-width of the underlying social graph?
\end{enumerate}

\subsubsection*{Acknowledgments:} Xin Huang was supported in part at the Technion Israel Institute of Technology by an Aly Kaufman Fellowship. A part of the work was done when Ganesh Ghalme was in Technion-Israel Institute of Technology. The authors thank Siddharth Barman and Ioannis Caragiannis for their helpful comments.

 \bibliography{ref}

\newcommand{\etalchar}[1]{$^{#1}$}
\begin{thebibliography}{CKM{\etalchar{+}}19}

\bibitem[ABC{\etalchar{+}}18]{aziz2018knowledge}
Haris Aziz, Sylvain Bouveret, Ioannis Caragiannis, Ira Giagkousi, and
  J{\'e}r{\^o}me Lang.
\newblock Knowledge, fairness, and social constraints.
\newblock In {\em Proceedings of the AAAI Conference on Artificial
  Intelligence}, volume~32, 2018.

\bibitem[ABKR19]{arunachaleswaran2019fair}
Eshwar~Ram Arunachaleswaran, Siddharth Barman, Rachitesh Kumar, and Nidhi
  Rathi.
\newblock Fair and efficient cake division with connected pieces.
\newblock In {\em International Conference on Web and Internet Economics},
  pages 57--70. Springer, 2019.

\bibitem[ACF{\etalchar{+}}18]{AmanatidisFMPV18}
Georgios Amanatidis, George Christodoulou, John Fearnley, Evangelos Markakis,
  Christos{-}Alexandros Psomas, and Eftychia Vakaliou.
\newblock An improved envy-free cake cutting protocol for four agents.
\newblock In {\em Algorithmic Game Theory - 11th International Symposium,
  {SAGT}'18}, volume 11059, pages 87--99. Springer, 2018.

\bibitem[Adj]{AdjustedWinner}
Adjusted winner.
\newblock \url{http://www.nyu.edu/projects/adjustedwinner/}.
\newblock Accessed: 2019-07-07.

\bibitem[AKP17]{AbebeKP16}
Rediet Abebe, Jon Kleinberg, and David~C. Parkes.
\newblock Fair division via social comparison.
\newblock AAMAS '17, 2017.

\bibitem[AM16a]{aziz2016discrete}
Haris Aziz and Simon Mackenzie.
\newblock A discrete and bounded envy-free cake cutting protocol for any number
  of agents.
\newblock In {\em 2016 IEEE 57th Annual Symposium on Foundations of Computer
  Science (FOCS)}, pages 416--427. IEEE, 2016.

\bibitem[AM16b]{aziz2016discretefour}
Haris Aziz and Simon Mackenzie.
\newblock A discrete and bounded envy-free cake cutting protocol for four
  agents.
\newblock In {\em Proceedings of the Forty-Eighth Annual ACM Symposium on
  Theory of Computing}, STOC '16, page 454–464, 2016.

\bibitem[BCE{\etalchar{+}}16]{brandt2016handbook}
Felix Brandt, Vincent Conitzer, Ulle Endriss, J{\'e}r{\^o}me Lang, and Ariel~D
  Procaccia.
\newblock {\em Handbook of computational social choice}.
\newblock Cambridge University Press, 2016.

\bibitem[BCG{\etalchar{+}}19]{beynier2019local}
Aur{\'e}lie Beynier, Yann Chevaleyre, Laurent Gourv{\`e}s, Ararat Harutyunyan,
  Julien Lesca, Nicolas Maudet, and Ana{\"e}lle Wilczynski.
\newblock Local envy-freeness in house allocation problems.
\newblock {\em Autonomous Agents and Multi-Agent Systems}, 33(5):591--627,
  2019.

\bibitem[BKN22a]{Bredereck18}
Robert Bredereck, Andrzej Kaczmarczyk, and Rolf Niedermeier.
\newblock Envy-free allocations respecting social networks.
\newblock {\em Artificial Intelligence}, 305:103664, 2022.

\bibitem[BKN22b]{bredereck2022envy}
Robert Bredereck, Andrzej Kaczmarczyk, and Rolf Niedermeier.
\newblock Envy-free allocations respecting social networks.
\newblock {\em Artificial Intelligence}, 305:103664, 2022.

\bibitem[BQZ17]{bei2017networked}
Xiaohui Bei, Youming Qiao, and Shengyu Zhang.
\newblock Networked fairness in cake cutting.
\newblock In {\em Proceedings of the Twenty-Sixth International Joint
  Conference on Artificial Intelligence, {IJCAI-17}}, pages 3632--3638, 2017.

\bibitem[BR21]{barman2021fair}
Siddharth Barman and Nidhi Rathi.
\newblock Fair cake division under monotone likelihood ratios.
\newblock {\em Mathematics of Operations Research}, 2021.

\bibitem[BSW{\etalchar{+}}20]{bei2020cake}
Xiaohui Bei, Xiaoming Sun, Hao Wu, Jialin Zhang, Zhijie Zhang, and Wei Zi.
\newblock Cake cutting on graphs: a discrete and bounded proportional protocol.
\newblock In {\em Proceedings of the Fourteenth Annual ACM-SIAM Symposium on
  Discrete Algorithms}, SODA '20, pages 2114--2123, 2020.

\bibitem[BT96]{brams1996fair}
Steven~J Brams and Alan~D Taylor.
\newblock {\em Fair Division: From cake-cutting to dispute resolution}.
\newblock Cambridge University Press, 1996.

\bibitem[Bud11]{budish2011combinatorial}
Eric Budish.
\newblock The combinatorial assignment problem: Approximate competitive
  equilibrium from equal incomes.
\newblock {\em Journal of Political Economy}, 119(6):1061--1103, 2011.

\bibitem[CEM17]{chevaleyre2017distributed}
Yann Chevaleyre, Ulle Endriss, and Nicolas Maudet.
\newblock Distributed fair allocation of indivisible goods.
\newblock {\em Artificial Intelligence}, 242:1--22, 2017.

\bibitem[CKM{\etalchar{+}}19]{procaccia_moulin_2016}
Ioannis Caragiannis, David Kurokawa, Herv\'{e} Moulin, Ariel~D. Procaccia,
  Nisarg Shah, and Junxing Wang.
\newblock The unreasonable fairness of maximum nash welfare.
\newblock 7(3), 2019.

\bibitem[DQS12]{deng2012algorithmic}
Xiaotie Deng, Qi~Qi, and Amin Saberi.
\newblock Algorithmic solutions for envy-free cake cutting.
\newblock {\em Operations Research}, 60(6):1461--1476, 2012.

\bibitem[DS61]{dubins1961cut}
Lester~E Dubins and Edwin~H Spanier.
\newblock How to cut a cake fairly.
\newblock {\em The American Mathematical Monthly}, 68(1P1):1--17, 1961.

\bibitem[EPT07]{etkin2007spectrum}
Raul Etkin, Abhay Parekh, and David Tse.
\newblock Spectrum sharing for unlicensed bands.
\newblock {\em IEEE Journal on selected areas in communications},
  25(3):517--528, 2007.

\bibitem[ES99]{edward1999rental}
Francis Edward~Su.
\newblock Rental harmony: Sperner's lemma in fair division.
\newblock {\em The American mathematical monthly}, 106(10):930--942, 1999.

\bibitem[Fol67]{foley1967resource}
Duncan~K Foley.
\newblock Resource allocation and the public sector.
\newblock 1967.

\bibitem[GMPZ16]{gal2016fairest}
Ya'akov Gal, Moshe Mash, Ariel~D Procaccia, and Yair Zick.
\newblock Which is the fairest (rent division) of them all?
\newblock In {\em Proceedings of the 2016 ACM Conference on Economics and
  Computation}, pages 67--84, 2016.

\bibitem[KLP13]{kurokawa2013cut}
David Kurokawa, John~K Lai, and Ariel~D Procaccia.
\newblock How to cut a cake before the party ends.
\newblock In {\em Twenty-Seventh AAAI Conference on Artificial Intelligence},
  2013.

\bibitem[Mou04]{moulin2004fair}
Herv{\'e} Moulin.
\newblock {\em Fair division and collective welfare}.
\newblock MIT press, 2004.

\bibitem[PM16]{procaccia2015cake}
Ariel~D. Procaccia and Hervé Moulin.
\newblock {\em Cake Cutting Algorithms}, page 311–330.
\newblock Cambridge University Press, 2016.

\bibitem[Pro09]{Procaccia2009ThouSC}
Ariel~D. Procaccia.
\newblock Thou shalt covet thy neighbor's cake.
\newblock In {\em Proceedings of the 21st International Joint Conference on
  Artificial Intelligence}, IJCAI'09, page 239–244, 2009.

\bibitem[RW98]{robertson1998cake}
Jack Robertson and William Webb.
\newblock {\em Cake-cutting algorithms: Be fair if you can}.
\newblock AK Peters/CRC Press, 1998.

\bibitem[Sim80]{simmons1980private}
FW~Simmons.
\newblock private communication to michael starbird.
\newblock 1980.

\bibitem[Ste48]{steihaus1948problem}
Hugo Steihaus.
\newblock The problem of fair division.
\newblock {\em Econometrica}, 16:101--104, 1948.

\bibitem[Str80]{stromquist1980cut}
Walter Stromquist.
\newblock How to cut a cake fairly.
\newblock {\em The American Mathematical Monthly}, 87(8):640--644, 1980.

\bibitem[Str08]{stromquist2008}
Walter Stromquist.
\newblock Envy-free cake divisions cannot be found by finite protocols.
\newblock {\em Electronic Journal of Combinatorics}, 15(1), 2008.

\bibitem[Tuc21]{Tucker21}
Jamie Tucker{-}Foltz.
\newblock Thou shalt covet the average of thy neighbors' cakes.
\newblock {\em CoRR}, abs/2106.11178, 2021.

\bibitem[Vos02]{vossen2002fair}
Thomas Vossen.
\newblock {\em Fair allocation concepts in air traffic management}.
\newblock PhD thesis, PhD thesis, Supervisor: MO Ball, University of Martyland,
  College Park, Md, 2002.

\end{thebibliography}
 \newpage 
\appendix
\section{Five agents on a \textsc{Line} graph}
\label{sec:5line} 

   

In this section, we present a simple protocol that finds a locally envy-free allocation among five agents on a \textsc{Line} using $18$ \cut \   queries and $29$ \eval \ queries (see Theorem \ref{thm:5Line}). We do so by extending the techniques developed for the case of four agents on a \textsc{Line} graph. 
Similar to $4$-agents case, the center agent (i.e. $a_3$)  is the designated cutter here as well. But in contrast with the $4$-agents case, the residue consists of trimmings from both left and right sides of $a_3$ i.e., both $a_2$ and $a_4$ act as trimmer agents.  The residue from both sides is then redistributed in multiple trimming  rounds until the cutter \emph{dominates} both her neighbouring agents. The dominance condition is defined formally in Section~\ref{sec:Linegraph}, but informally, it says that the value of $a_3$ for the  residue created in a given trimming round is so small that it does not induce envy from $a_3$ towards her neighbors irrespective of how the bundles are allocated between them.

Once the \emph{dominance} is achieved---which happens in at most two rounds of the \texttt{Trimming} phase---the current residue can be distributed without creating local envy in the \texttt{Equaling} phase. The distinction from the $4$-agents case (which has a single trimmer agent) is that it requires an additional round of \texttt{Trimming} to create \emph{dominance} of the cutter agent over the two trimmer agents.  We state our result here. 

\begin{corollary}
For cake-division instances with $5$ agents on  a \textsc{Line} graph, there exists a discrete cake cutting protocol that finds a locally envy-free allocation using $18$ \cut \  queries and $29$ \eval \ queries.
\label{thm:5Line}
\end{corollary}

\begin{proof}
We begin by describing the protocol, followed by the analysis for its correctness and query complexity. As mentioned earlier, this protocol is an extension of \textsc{Alg1}, and hence we will a short description while detailing out the key differences.
\newline \noindent 
\textbf{The Protocol:} A locally envy-free  protocol for five agents on a \textsc{Line} graph consists of the center agent ($a_3$) as the cutter and her  neighbours  ($a_2$ and $a_4$) as  trimmer agents. It begins with the cutter agent dividing the cake into $5$ equal pieces using \textsc{Eq-Div}, followed by $a_2$ and $a_4$ each picking their two highest-valued (available) pieces one after the other. The remaining last piece in added to the cutter agent's bundle.  Both the trimmer agents then perform  \textsc{Trim}(.) procedure in the first  round (trimming phase). After this round, the cutter agent starts to \emph{dominate}  at least one trimmer agent, say $a_2$ (see Claim \ref{claim:5line}), and therefore $a_2$ performs the   \textsc{Equal}(.) procedure for the remaining rounds. In the next round, $a_4$ continues to perform the  \textsc{Trim}(.) procedure. Finally, in the third round, we show that the dominance is established over $a_4$ as well, and both $a_2$ and $a_4$ perform \textsc{Equal}(.) procedure in the last round. 

As opposed to the $4$-agents case, an additional trimming round is attributed to the fact that trimming comes from  two trimmer agents in this case. The final bundles $\{A_1, A_2, A_3, A_4, A_5\}$ are created similar to Step \ref{alg1:group} of \textsc{Alg1}.
The cutter agent is allocated a bundle $A_3$ that consists of the remaining pieces from all the rounds (i.e. both \texttt{Trimming} and \texttt{Equaling} phases) after $a_2$ and $a_4$ make their selections from \textsc{Eq-Div}.
 Two---one from each side of the cutter agent---of the other four bundles, $A_2$ and $A_4$, contain  untrimmed pieces from the \texttt{Trimming} phase and a  trimmed piece from the \texttt{Equaling} phase. This ensures that $a_3$ values his bundle $A_3$ as high as $A_2$ and $A_4$.
 The remaining two bundles (again, one from each side of cutter agent) $A_1$ and $A_5$ contain trimmed pieces from the \texttt{Trimming} phase and appended piece from the \texttt{Equaling} phase. The \emph{dominance} of the cutter agent $a_3$ over $a_2$ and $a_4$ further ensures that $a_3$ will not have any envy for these bundles as well. 
 
  As mentioned earlier, $a_3$ is allocated bundle $A_3$. Agent $a_1$ picks her favorite bundle between $A_1$ and $A_2$, and the remaining bundle is allocated to agent $a_2$, and similarly $a_5$ pick her favourite bundle between $A_4$ and $A_5$, and the remaining bundle goes to agent $a_4$ (similar to Step 11 of \textsc{Alg1}). Recall that, \texttt{Trimming} and \texttt{Equaling} phases always create equi-valued bundles for $a_2$ and $a_4$, and that is valued at least as high as the bundle $A_3$, the protocol ensures local envy-freeness.\\

\noindent 
\textbf{Correctness}: We now show that the above algorithm returns a locally envy-free allocation of the cake.  Recall that there are three rounds in total, in the first round both trimmer agents $a_2$ and $a_4$ perform \textsc{Trim}(.) procedure, in the second round at least one of $a_2$ or $a_4$ performs \textsc{Trim}(.) and the other performs \textsc{Equal}(.) procedure, while in the last round, both trimmer agents perform \textsc{Equal}(.) procedure.
Let us assume the cutter agent $a_3$ divides the current residue into five equal pieces $\{P^j_1,P^j_2, \dots, P^j_5\}$ in round $j=\{1,2,3\}$. 
Also write $\{A^j_1,A^j_2, \dots, A^j_5\}$ to denote the partial partition of the cake obtained at the end of round $j = \{1,2,3\}$. Next, we denote the trimmings in the first round as $T_2$ and $T_4$ from agent $a_2$ and $a_4$ respectively. We say that $T=T_2 \cup T_4$ is the total residue of the first round. In the second round, since there is only one agent who performs \textsc{Trim}(.), we denote this residue by $T'$.

Note that we have $\cup_{i \in [5]} P^1_i= [0,1]$ in the first round, $\cup_{i \in [5]} P^2_i= T$ in the second round, and $\cup_{i \in [5]} P^3_i= T'$ in the third round. Without loss of generality, let us assume $a_2$ picks $P^j_1$ and $P^j_2$, and $a_4$ picks $P^j_4$ and $P^j_5$ in round $j$ in the \textsc{Select}(.) procedure. 
The following claim proves a crucial property of our algorithm, which would be formally defined as \emph{dominance condition} (see Definition \ref{def:domination})  when we discuss about \dTwoTree \ in Section \ref{sec:depth2}. 

\begin{claim} \label{claim:5line}
Consider the partial partition $\{A^1_1,A^1_2, \dots, A^1_5\}$ of the cake at the end of the first round of \textsc{Trim}(.) along-with trimming $T$. Then, at least one of the following two conditions must hold true.
\begin{enumerate}
    \item [a)] $v_{3}(A^1_3) - v_3(A^1_2) \geq \frac{2}{5} \cdot  v_3(T) $ 
    \item [b)]  $v_3(A^1_3) - v_3 (A^1_4) \geq \frac{2}{5} \cdot v_3(T)$. 
\end{enumerate}
\end{claim}

\begin{proof}
Let us assume for contradiction that both of the stated conditions are not satisfied at the end of the first round. That is, we have  $v_{3}(A^1_3) - v_3(A^1_2) < \frac{2}{5} \cdot  v_3(T) $ 
  and  $v_3(A^1_3) - v_3 (A^1_4) < \frac{2}{5} \cdot v_3(T)$. Summing these two inequalities, we obtain
    $$2 v_{3}(A^1_3) - v_3(A^1_2) - v_3(A^1_4) < \frac{4}{5} \cdot  v_3(T)$$
    
    Recall that $A^1_2 = P^1_2 \setminus T_2$ and $A^1_4 = P^1_4 \setminus T_4$ after the first \textsc{Trimming} round. Therefore, by additivity of valuations, we obtain
     $$2 v_{3}(A^1_3) - (v_3(P^1_2)-v_3(T_2))- (v_3(P^1_4)-v_3(T_4)) < \frac{4}{5} \cdot  v_3(T)$$
     
     Since $a_3$ performs the \textsc{Equal}(.) procedure, and $A^1_3 = P^1_3$ is of value $1/5$ for $a_3$, we have
    $$\frac{2}{5} - (\frac{1}{5}-v_3(T_2)) - (\frac{1}{5}-v_3(T_4)) < \frac{4}{5} \cdot  v_3(T)$$
    
    Finally, since $T=T_2 \cup T_4$, by additivity of valuations, we obtain $v_3(T) < \frac{4}{5} \cdot  v_3(T)$, a contradiction. Therefore, either one of the conditions holds true at the end of the first round, proving the stated claim.
\end{proof}

If condition $(a)$ is true, we say that that the cutter agent \emph{dominates} agent $a_3$, otherwise we say that $a_2$ \emph{dominates} agent $a_4$. If condition $(a)$ is satisfied, it essentially says that the cutter agent will not envy the bundle $A^1_2$, even if she adds $2/5$th of the residue $T$ to it in the next round. Therefore, she proceeds to perform \textsc{Equal}(.) procedure.

Following the similar arguments, we show that the cutter agent will start dominating the remaining trimmer agent after one additional round of \textsc{Trimming}. Hence, by the above description of the protocol and the ideas used in proving Theorem \ref{thm:4LINE}, we establish the fact that \textsc{Alg1} can be extended to achieve local envy-freeness among five agents on a \textsc{Line} graph.\\

\noindent 
\textbf{Counting Queries:} It is easy to count the number of queries required in this protocol. Due to similar arguements, we urge the readers to refer the $4$-agents case for details, here we state the number of queries required in each round for completeness. The first round requires $6$ \texttt{cut} queries and $7$ \texttt{eval} queries, the second and third rounds each require $6$ \texttt{cut} queries and $9$ \texttt{eval} queries. Finally, it requires $4$ \texttt{eval} queries to allocate the bundles at the end.
\end{proof}

\end{document}